\documentclass[12pt]{article}

 \usepackage[T1]{fontenc}              
 \usepackage[width=16cm, height=22cm]{geometry}

\usepackage{amsmath,amssymb,bm,color,mathrsfs,mathtools}
 \usepackage[amsmath,thmmarks,hyperref]{ntheorem} 
 \usepackage{hyperref} 
\usepackage{amscd,verbatim}
\usepackage{bbding,wasysym,pifont}
\usepackage{graphicx}
\usepackage{subfigure}
\usepackage{multirow}
\usepackage{comment}
\usepackage{todonotes}
\usepackage[arrow,matrix,curve]{xy}
\usepackage{subfigure}
\usepackage{authblk}
\usepackage{tikz-cd}
\usepackage{colortbl}
\usepackage{rotating}
\usepackage{relsize}
 \usepackage[numbers,square]{natbib}
 \usepackage{slashed}
 \usepackage{xy}
 \usepackage{url}
\usepackage{tikz-cd}
\usetikzlibrary{decorations.markings}
\tikzset{->-/.style={decoration={
  markings,
  mark=at position #1 with {\arrow{>}}},postaction={decorate}}}
\tikzset{->-/.default=0.5}

\usepackage{pgfplots}

\pgfplotsset{compat=1.10}
\usepgfplotslibrary{fillbetween}
\usetikzlibrary{patterns}

\newcommand{\CC}{{\mathbb C}}
\newcommand{\NN}{{\mathbb N}}
\newcommand{\RR}{{\mathbb R}}

\newcommand{\ZZ}{{\mathbb Z}}
\newcommand{\HH}{{\mathbb H}}

\newcommand{\EE}{{\mathbb E}}

\theoremstyle{nonumberplain}  
\theoremheaderfont{\itshape}  
\theorembodyfont{\normalfont}  
\theoremseparator{.}  
\theoremsymbol{$\Box$}
\newtheorem{proof}{Proof} 
  
\theoremstyle{plain}  
\theoremheaderfont{\normalfont\bfseries}  
\theorembodyfont{\itshape}  
\theoremsymbol{~}
\theoremseparator{.}  
\newtheorem{proposition}{Proposition}[section]  
  
\newtheorem{corollary}[proposition]{Corollary}  
\newtheorem{lemma}[proposition]{Lemma}  
\newtheorem{theorem}{Theorem}   

\theorembodyfont{\normalfont}  
 
\newtheorem{remark}[proposition]{Remark}

\theoremstyle{nonumberplain}
\theoremheaderfont{\normalfont\bfseries}  
\theorembodyfont{\itshape}  
\theoremsymbol{~}
\theoremseparator{.}

\begin{document}

\title{Gaplessness of Landau Hamiltonians on hyperbolic half-planes via coarse geometry}

\author[1]{Matthias Ludewig}
\author[2]{Guo Chuan Thiang}
\affil[1]{University of Regensburg, Germany}
\affil[2]{Beijing International Center for Mathematical Research, Peking University, China}

%\date{today}

\maketitle

\begin{abstract}
We use coarse index methods to prove that the Landau Hamiltonian on the hyperbolic half-plane, and even on much more general imperfect half-spaces, has no spectral gaps. Thus the edge states of hyperbolic quantum Hall Hamiltonians completely fill up the gaps between Landau levels, just like those of the Euclidean counterpart. 
  \end{abstract}

\section*{Introduction}

Let $X$ be either the Euclidean plane $\EE$ or the hyperbolic plane $\HH$. By a uniform magnetic field of strength $\theta\in\RR$ perpendicular to $X$, we mean the closed two-form $F_\theta=\theta\cdot\omega$, where $\omega$ denotes the normalised invariant (under the isometry group) volume form on $X$. Let $\mathcal{L}_\theta$ be the trivial Hermitian line bundle $\mathcal{L}_\theta=X\times\CC$ with connection 1-form\footnote{While the gauge group ${\rm U}(1)$ has Lie algebra $\mathfrak{u}(1)\cong i\RR$, it is customary in physics to use a real connection 1-form. There is also a sign choice depending on the electric charge $q$, which enters in the formula $d-iqA_\theta$ for the covariant derivative (``minimal coupling''). With suitable units, $q=-1$ for an electron, for instance.} $A_\theta$ and curvature $dA_\theta=F_\theta$. The \emph{Landau Hamiltonian} $H_\theta$ is the connection Laplacian on $\mathcal{L}_\theta$,
\begin{equation}
H_\theta:=(d-iA_\theta)^*(d-iA_\theta).\label{eqn:Landau.Hamiltonian}
\end{equation}
A different choice of $A_\theta$, or \emph{gauge}, gives rise to a unitarily equivalent $H_\theta$ with the same spectrum.

For $\theta =0$, we recover the standard Laplacian $H_0 = \Delta$, the spectrum of which is well-known to be $[0,\infty)$ in the Euclidean, respectively $[\frac{1}{4},\infty)$ in the hyperbolic case, see \cite{McKean}.
For non-zero $\theta$, the spectrum of $H_\theta$ differs dramatically from that of $H_0$: In the Euclidean case, the spectrum of $H_\theta$ is an infinite sequence
\begin{equation*}
\lambda_{m,\theta}=(2m+1)|\theta|,\qquad m=0,1,\ldots,
\end{equation*}
of infinitely-degenerate isolated eigenvalues, called \emph{Landau levels} \cite{Landau} (see Fig.\ \ref{fig:euc.raising.lowering}). In the hyperbolic plane case, $H_\theta$ has a finite sequence of eigenvalues
\begin{equation*}
\lambda_{m,\theta}=(2m+1)|\theta |-m(m+1),\quad m=0,1,\ldots m_{\rm max}< |\theta |-\frac{1}{2},
\end{equation*}
as well as a continuous spectrum $[\frac{1}{4}+\theta^2,\infty)$, see \cite{Comtet}. Thus some isolated Landau levels occur once $|\theta|>\frac{1}{2}$ (see Fig.\ \ref{fig:hyp.raising.lowering}).
 
\medskip

Our main result is the following.

\begin{theorem}\label{thm:main.theorem.intro}
Let $X$ be the hyperbolic or Euclidean plane, $W$ be the closed half-plane lying on one side of a geodesic. Let $H_{\theta, W}$ be a Landau Hamiltonian on $W$, defined by either Dirichlet or Neumann boundary conditions. Then the spectrum of $H_{\theta,W}$ has no gaps above the lowest Landau level $\lambda_{0,\theta}=|\theta |$. 
\end{theorem}
 
Thus, the Landau Hamiltonian $H_\theta$ exhibits the \emph{gap-filling phenomenon}: all the spectral gaps between its Landau levels (and the continuous spectrum in the hyperbolic case) get completely filled up after passing to its half-plane version $H_{\theta,W}$. 

While  for simplicity, we formulated Thm.~\ref{thm:main.theorem.intro} for the case that $\partial W$ is a geodesic, a similar statement also holds for ``imperfect'' half-spaces $W$ which are, for instance, merely quasi-isometric to a standard half-plane, see Remark \ref{RemarkImperfekt}. More general boundary conditions can be treated as well; see \S\ref{SectionMainTheorem} for a more general version of Thm.~\ref{thm:main.theorem.intro}. Our methods are also applicable to higher-dimensional $X$; see Remark \ref{rem:higher.dimensions}.
 
%more general boundary conditions can be treated as well, see \S\ref{SectionGapFilling} for a precise statement.

In the Euclidean case, we showed in \cite{LTcobordism} that Theorem \ref{thm:main.theorem.intro} holds (also for more general $W$), because the Landau Hamiltonian encounters a certain equivariant coarse index obstruction to maintaining the spectral gaps between its Landau levels. It is also possible to directly compute the Dirichlet and Neumann spectrum for $H_{\theta,W}$ and verify the gap-filling phenomenon \cite{Pule,BMR}. However, a direct spectral computation for $H_{\theta,W}$ in the hyperbolic case would be much more difficult, if at all possible.

For this paper, we use a conceptually simpler \emph{nonequivariant} coarse index obstruction, which serves the same purpose as far as demonstrating gap-filling is concerned (but see Remark \ref{rem:equivariant.advantage}). We compute that this obstruction occurs for both the Euclidean and hyperbolic plane Landau Hamiltonian, thereby proving Theorem \ref{thm:main.theorem.intro} \emph{without having to directly solve the spectral problem for $H_{\theta,W}$, or to assume special boundary geometries and conditions}. 

\medskip

Our result is motivated by physicists' investigations of so-called \emph{topological phases}. The intuition is that each Landau level is in some sense ``topologically non-trivial'', and leaves a robust signature on the material boundary through gap-filling and boundary-localised states via the \emph{bulk-boundary correspondence} principle. This principle has been theoretically and experimentally verified for many physical systems in Euclidean geometry. The possibility of quantum Hall effects and other topological phases on $\HH$ was contemplated in \cite{Comtet, CHMM, MThyperbolic}. However, it has remained an open question, even for the paradigmatic Landau Hamiltonian, whether the gap-filling phenomenon occurs in hyperbolic (or more general) geometries. With the ability to effectively simulate dynamics in hyperbolic geometry \cite{QuantumSimulation, Kollar}, this has become a pertinent question to address, and our Theorem \ref{thm:main.theorem.intro} answers this in the affirmative.

\medskip

\noindent {\bf Outline.} 
In \S\ref{sec:Roe.algebras}, we construct an exact sequence of Roe algebras, corresponding to operators $H_X$ acting on a manifold $X$ and $H_W$ acting on a subset $W\subset X$. We explain how spectral gaps of $H_X$ encounter gap-filling in the passage to $H_W$, whenever the spectral projection encounters a coarse index obstruction. In \S\ref{sec:Landau.Dirac}, we explain the geometric origin of the relationship between Landau Hamiltonians and twisted Dirac operators, and show that the Landau level eigenspaces are kernels of the latter. In \S\ref{sec:coarse.Landau}, we show that the Landau spectral projections realise the non-vanishing coarse index of the twisted Dirac operator. This coarse index forces the spectral gap above each Landau level to be filled when passing from $H_X=H_\theta$ to $H_W=H_{\theta,W}$.

\section{Exact sequences of Roe algebras associated to subsets of manifolds}\label{sec:Roe.algebras}

Let $X$ be a Riemannian manifold, and let $W \subset X$ be a regular closed subset, meaning that it is equal to the closure of its interior. In this section, we will construct a \emph{boundary map} in $K$-theory,
\begin{equation*}
  \partial_i: K_i(C^*(X)) \longrightarrow K_{i-1}(C^*(\partial W)),
\end{equation*}
which connects the $K$-theory of the Roe algebra of $X$ (``bulk Roe algebra''), to the $K$-theory of the Roe algebra of $\partial W$ (``boundary Roe algebra''). It turns out that there are various descriptions for this map.

We start with a review of Roe algebras. A bounded operator $T \in \mathcal{B}(L^2(X))$ is called of \emph{finite propagation} if there exists $R>0$ such that $\mathrm{supp}(T\varphi)$ is contained in the $R$-ball around $\mathrm{supp}(\varphi)$, for all $\varphi \in C_0(X)$, the space of continuous functions on $X$ vanishing at infinity. The norm closure in $\mathcal{B}(L^2(X))$ of the algebra of all pseudolocal\footnote{An operator $T$ is pseudolocal if $[\varphi, T]$ is compact for all $\varphi \in C_0(X)$, see \S5 of \cite{HRcoarse}.} and finite propagation operators is denoted by $D^*(X)$.
$T$ is called \emph{locally compact} if $T\varphi$ and $\varphi T$ are compact operators for every $\varphi \in C_0(X)$. The Roe algebra of $X$ can then be defined as the $C^*$-algebra
\begin{equation*}
  C^*(X) := \overline{\{T \in D^*(X) \mid T ~\text{locally compact}\}} \subseteq D^*(X) \subseteq \mathcal{B}(L^2(X)).
\end{equation*}
Let $S \subset X$ be an arbitrary closed subset. An operator $T \in \mathcal{B}(L^2(X))$ is \emph{supported near} $S$ if there exists $R>0$ such that $\varphi T = T \varphi = 0$ whenever $\varphi \in C_0(X)$ with $\mathrm{dist}(\mathrm{supp}(\varphi), S) \geq R$. The \emph{localized Roe algebra} $C^*_X(S)$ is the norm closure
\begin{equation*}
  C^*_X(S) := \overline{\{ T \in C^*(X) \mid T ~\text{supported near}~S  \}} \subseteq C^*(X) \subseteq \mathcal{B}(L^2(X)).
\end{equation*}
It follows from coarse invariance of Roe algebras \cite[Thm.~2.9]{EwertMeyer} and continuity of $K$-theory that always $K_i(C^*_X(S)) \cong K_i(C^*(S))$. However, the advantage of using the localized Roe algebra over $C^*(S)$ is that $C^*_X(S)$ is a two-sided ideal  in $C^*(X)$, hence leads to a six-term exact sequence in $K$-theory.

\begin{remark}
  One frequently encounters operators that act on sections of a vector bundle $\mathcal{V}$ on $X$ instead of scalar functions, such as the Dirac operator. For those operators, one can replace the space $L^2(X)$ by the space $L^2(X, \mathcal{V})$ of sections. More generally, it is often convenient to replace $L^2(X)$ by an abstract Hilbert space $\mathcal{H}$ together with a $*$-representation of $C_0(X)$, for which a Roe algebra can be defined analogously. It turns out that any such choice of $\mathcal{H}$ gives rise to an isomorphic Roe algebra, provided the representation is \emph{ample}, meaning that no non-zero $f \in C_0(X)$ acts as a compact operator (this condition is satisfied for $L^2$ spaces on manifolds because non-trivial multiplication operators are never compact). Moreover, the $K$-theory groups of Roe algebras corresponding to two ample $C_0(X)$-modules $\mathcal{H}$, $\mathcal{H}^\prime$ are \emph{canonically} isomorphic \cite[Thm.~2.1]{EwertMeyer}. However, for our purposes, it suffices to just always take $\mathcal{H} = L^2(X)$.
\end{remark}

\subsection{The localization sequence} \label{SectionLocalizationSequence}

Let $X$ be a Riemannian manifold, and let $W \subset X$ be a regular closed subset. Using the above definitions for $W$ and $\partial W$, we obtain the short exact sequence
\begin{equation} \label{ShortExactW}
  0 \longrightarrow C^*_W(\partial W) \longrightarrow C^*(W) \longrightarrow C^*(W)/C^*_W(\partial W) \longrightarrow 0
\end{equation}
and the corresponding exact six-term sequence
\begin{equation} \label{SixTerm1}
\begin{tikzcd}
  K_0(C^*_W(\partial W)) \ar[r] & K_0(C^*(W)) \ar[r] & K_0(C^*(W)/C^*_W(\partial W)) \ar[d, "\delta_1"] \\
  K_1(C^*(W)/C^*_W(\partial W)) \ar[u, "\delta_0"] & \ar[l] K_1(C^*(W)) & \ar[l] K_1(C^*_W(\partial W))
\end{tikzcd}
\end{equation}
in $K$-theory. We remark that above, we defined the Roe algebra of a Riemannian manifold only, while $W$ is not a manifold in general, but the same definitions work in this slightly more general case, since $W$ is regular, hence inherits a non-degenerate measure by restriction.

Let $\Pi: L^2(X) \rightarrow L^2(W)$ be the orthogonal projection, and $\Pi^*:L^2(W)\rightarrow L^2(X)$ the inclusion. There are canonical maps $e: C^*(W) \rightarrow C^*(X)$ and $r: C^*(X) \rightarrow C^*(W)$ given by $T  \mapsto \Pi^* T \Pi$ and  $T \mapsto \Pi T \Pi^*$, respectively. Since $\Pi \Pi^*$ is the identity on $L^2(W)$, the \emph{extension-by-zero map} $e$ is a $*$-homomorphism. In contrast, the \emph{restriction map} $r$ is not: it is $*$-preserving,  $r(T^*) = r(T)^*$, but it is not multiplicative in general. Nevertheless, we have the following lemma.

\begin{lemma} \label{LemmaRestrictionMap}
For $T, T^\prime\in C^*(X)$, we have $r(TT^\prime) - r(T)r(T^\prime) \in C^*_W(\partial W)$, hence the restriction map $r$ descends to a $*$-homomorphism $\tilde{r}: C^*(X) \rightarrow C^*(W)/C^*_W(\partial W)$.
\end{lemma} 

\begin{proof}
If the propagation of $T^\prime$ is less than $R>0$, then for all $\varphi \in C_0(W)$ such that $\mathrm{dist}(\mathrm{supp}(\varphi), \partial W) \geq R$, the composition $(\mathrm{id}-\Pi^*\Pi) T^\prime \varphi$ is zero. Therefore, under the same assumption,
\begin{equation*}
(r(TT^\prime) - r(T)r(T^\prime))\varphi = (\Pi T T^\prime \Pi^* - \Pi T\Pi^*\Pi T^\prime \Pi^*)\varphi = \Pi T(\mathrm{id}-\Pi^*\Pi) T^\prime\Pi^* \varphi = 0.
\end{equation*}
Taking adjoints, the same argument shows that $\varphi(r(TT^\prime) - r(T)r(T^\prime)) = 0$, provided that also the propagation of $T$ is less than $R$. Hence $r(TT^\prime) - r(T)r(T^\prime) \in C^*(W)$ is supported near $\partial W$ for finite propagation operators $T$ and $T^\prime$. For general $T$ and $T^\prime$, the statement follows from continuity of $r$, since finite propagation operators are dense in $C^*(X)$.
\end{proof}

Precomposing the exponential map of the six-term sequence \eqref{SixTerm1} with the map in $K$-theory induced by $\tilde{r}$, we get a homomorphism
\begin{equation} \label{BoundaryMap1}
  \partial_i := \delta_i \circ \tilde{r}_*: K_i(C^*(X)) \longrightarrow K_{i+1}(C^*_W(\partial W)) \cong K_{i+1}(C^*(\partial W)).
\end{equation}
In the next section, we will identify this map with the boundary map of a Mayer--Vietoris sequence, under suitable conditions.

\subsection{The Mayer--Vietoris sequence}

Let $X$ be a Riemannian manifold and let $W, W^\prime$ be a partition of $X$, by which we mean two regular closed subsets of $X$ such that $W \cup W^\prime = X$ and $W \cap W^\prime = \partial W$. We say that the partition is \emph{coarsely excisive}, if for all $R>0$, there exists $S>0$ such that if $x \in X$ satisfies both $d(x, W) \leq R$ and $d(x, W^\prime) \leq R$, then $d(x, \partial W) \leq S$, see Fig.\ \ref{fig:partitions} for some examples.

\begin{figure}
\begin{tikzpicture}
\draw[name path=A] (1,0) --(1,4);
\draw[white, name path =B] (3,0) -- (3,4);
\tikzfillbetween[of=A and B]{opacity=0.2};
\node at (0.3,2) {$W^\prime$};
\node at (2,2) {$W$};

\draw[name path=A] (4,2.5) .. controls (5,2.5) .. (5,4) ;
\draw[name path=B] (5.5,0) .. controls (5.5,1.5) .. (6.5,1.5);
\draw[white, name path=C] (4,0) -- (5.5,0);
\draw[white, name path=D] (5,4) -- (6.5,4);
\tikzfillbetween[of=D and B]{opacity=0.2};
\tikzfillbetween[of=A and C]{opacity=0.2};
\node at (4.3,3.5) {$W^\prime$};
\node at (6.2,0.5) {$W^\prime$};
\node at (5,2) {$W$};

\draw[name path=A] (11,3) -- (9,3).. controls (7,3.1) and (7,0.9) ..(9,1) -- (11,1);
\draw[white, name path =B] (11,0) -- (11,4);
\tikzfillbetween[of=A and B]{opacity=0.2};
\node at (8,3.3) {$W^\prime$};
\node at (9,2.5) {$W$};

\end{tikzpicture}~~~~~
\begin{tikzpicture}
\draw[name path=A] (2,0) --(2,2)--(3,2)--(3,4);
\draw[white, name path =B] (4,0) -- (4,4);
\draw[name path=C] (1,4) --(1,2)--(2,2);
\draw (1,3.8)--(1.2,4);
\draw (1,3.6)--(1.4,4);
\draw (1,3.4)--(1.6,4);
\draw (1,3.2)--(1.8,4);
\draw (1,3.0)--(2.0,4);
\draw (1,2.8)--(2.2,4);
\draw (1,2.6)--(2.4,4);
\draw (1,2.4)--(2.6,4);
\draw (1,2.2)--(2.8,4);
\draw (1,2)--(3,4);
\draw (1.2,2)--(3,3.8);
\draw (1.4,2)--(3,3.6);
\draw (1.6,2)--(3,3.4);
\draw (1.8,2)--(3,3.2);
\draw (2,2)--(3,3);
\draw (2.2,2)--(3,2.8);
\draw (2.4,2)--(3,2.6);
\draw (2.6,2)--(3,2.4);
\draw (2.8,2)--(3,2.2);
\tikzfillbetween[of=A and B]{opacity=0.2};
\node at (1,1) {$W^\prime$};
\node at (3,1) {$W$};

\end{tikzpicture}
\caption{The first diagram shows the standard partition of the plane $X$ by the half-plane $W$ on one side of a geodesic (for $X$ the hyperbolic plane, we are considering $y>0$). Other possible coarsely excisive partitions $X=W\cup W^\prime$ are shown in the second and third diagram. However, the third partition does not satisfy the condition that the distance function to $W\cap W^\prime=\partial W$ is unbounded. The fourth diagram depicts a non-coarsely excisive partition of a subset of the plane (the hatched area is not part of the subset).}\label{fig:partitions}
\end{figure}
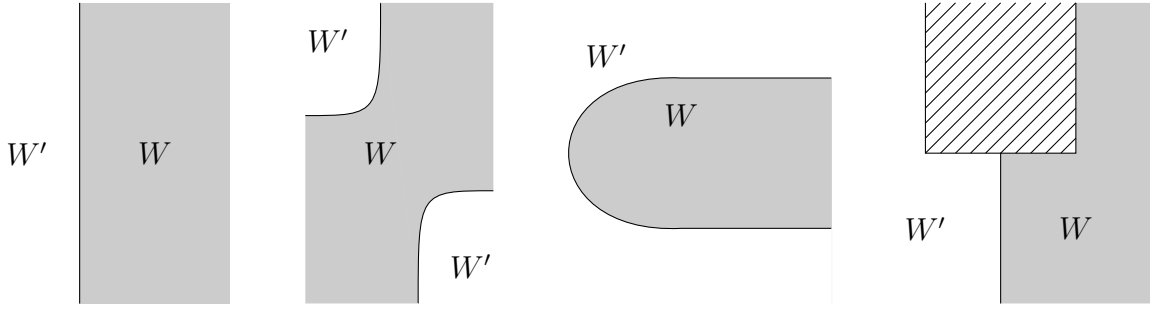

Let $i$ and $i^\prime$ be the inclusions of $C^*_X(\partial W)$ into $C^*_X(W)$ and $C^*_X(W^\prime)$, and let $j$, $j^\prime$ be the further inclusions into $C^*(X)$. Then we have the following Mayer--Vietoris exact sequence, see \cite[Corollary~3]{EwertMeyer} and \cite[\S5]{HRY},
\begin{equation} \label{MayerVietoris}
\begin{tikzcd}
K_0(C^*_X(\partial W)) \ar[r, "(i_*{,} i^\prime_*)"] & K_0(C^*_X(W)) \oplus K_0(C^*_X(W^\prime)) \ar[r, "j_* - j^\prime_*"] & K_0(C^*(X)) \ar[d, "\partial_0"] \\
K_1(C^*(X)) \ar[u, "\partial_1"] & \ar[l, "j_* - j^\prime_*"] K_1(C^*_X(W)) \oplus K_1(C^*_X(W^\prime)) &\ar[l, "(i_*{,} i^\prime_*)"] K_1(C^*_X(\partial W))
\end{tikzcd}
\end{equation}
The boundary maps $\partial_i$ are given as follows.

\begin{proposition}
Suppose that the partition $W$, $W^\prime$ of $X$ is coarsely excisive.
Then after identifying $K_i(C^*_W(\partial W)) \cong K_i(C^*_X(\partial W))$ with the extension-by-zero map $e_*$, the boundary maps of the Mayer--Vietoris sequence \eqref{MayerVietoris} coincide with the maps $\partial_i$ defined in \eqref{BoundaryMap1}. 
\end{proposition}

\begin{proof}
By \cite[Prop.~6 \& Eq.~(8)]{EwertMeyer}, the boundary map $\partial_i$ in the Mayer--Vietoris sequence is given as the composition of the left vertical column in the diagram below.
\begin{equation*}
\begin{tikzcd}
K_i(C^*(X)) \ar[d] \ar[ddd, bend right=90, "\partial_i"'] \ar[ddr, "\tilde{r}_*", bend left=30]\\
K_i(C^*(X)/C^*_X(W^\prime)) \\
K_i(C^*_X(W)/C^*_X(\partial W)) \ar[u, "\cong"] \ar[d, "\delta_i"] & \ar[l, "e_*"', "\cong"] K_i(C^*(W)/C^*_W(\partial W)) \ar[d, "\delta_i"]\\
 K_{i-1}(C^*_X(\partial W)) & \ar[l, "e_*", "\cong"'] K_{i-1}(C^*_W(\partial W))
\end{tikzcd}
\end{equation*}
The diagram is easily checked to be commutative and the desired equality follows from the commutativity of the outmost square of the diagram.
\end{proof}

\subsection{Relation to quasi-equivariant exact sequence}
For this subsection, suppose that the Riemannian manifold $X$ carries an effective, cocompact, properly discontinuous, isometric action of a discrete countable group $\Gamma$. In this case, one can consider the \emph{equivariant} Roe algebra $C^*(X, \Gamma) \subseteq C^*(X)$, which is the norm closure of all locally compact, finite propagation, $\Gamma$-invariant operators in $\mathcal{B}(L^2(X))$. 
Assume that $W \subset X$ is a closed subset such that the distance function to $\partial W$ is unbounded (see Fig.\ \ref{fig:partitions}).
In \cite{LTcobordism}, it is shown that there is a short exact sequence (compare Eq.\ \eqref{ShortExactW})
\begin{equation} \label{ShortExactQuasi}
  0 \longrightarrow C^*_W(\partial W) \longrightarrow Q^*(W, \Gamma) \longrightarrow C^*(X, \Gamma) \longrightarrow 0,
\end{equation}
where $Q^*(W, \Gamma)$ is a certain algebra of \emph{quasi-invariant} operators. Concretely, if $U_\gamma$ are the unitary operators on $L^2(X)$ representing the elements $\gamma \in \Gamma$ and $V_\gamma := \Pi U_\gamma \Pi^*$ their compression on $L^2(W)$, the algebra $Q^*(W, \Gamma)$ is the set of those operators $T \in C^*(W)$ such that $V_\gamma T - T V_\gamma \in C^*_W(\partial W)$ for all $\gamma \in \Gamma$. 

\begin{proposition}
The boundary map of the $K$-theory six-term sequence associated to \eqref{ShortExactQuasi} is just the boundary map \eqref{BoundaryMap1}, precomposed with the canonical map $K_i(C^*(X, \Gamma)) \rightarrow K_i(C^*(X))$ that forgets equivariance.
\end{proposition}

\begin{proof}
We have the following diagram of short exact sequences, where the right vertical map is the restriction map from Lemma~\ref{LemmaRestrictionMap} (restricted to the equivariant subalgebra) and the middle vertical map is just the inclusion.
\begin{equation*}
\begin{tikzcd}
  0 \ar[r] & C^*_W(\partial W) \ar[d, equal] \ar[r] & Q^*(W, \Gamma) \ar[d] \ar[r] & C^*(X, \Gamma) \ar[d, "\tilde{r}"] \ar[r] & 0 \\
  0 \ar[r] & C^*_W(\partial W) \ar[r] & C^*(W) \ar[r] & C^*(W)/ C^*_W(\partial W) \ar[r] & 0
\end{tikzcd}
\end{equation*}
The proposition follows from the naturality of the $K$-theory boundary maps.
\end{proof}

\begin{remark}
The condition that  \emph{the function $x \mapsto \mathrm{dist}(x, \partial W)$ is unbounded on $W$} is needed to construct the right map in the short exact sequence \eqref{ShortExactQuasi}. Notice that if this condition is violated, the short exact sequence \eqref{ShortExactW} is trivial, because $C^*(W)/C^*_W(\partial W) = 0$ if all of $W$ is within finite distance of $\partial W$. Consequently, the Mayer--Vietoris sequence \eqref{MayerVietoris} becomes trivial in the sense that the boundary maps vanish, as follows from the fact that then $i : C^*_X(\partial W) \rightarrow C^*_X(W)$ is an isomorphism. 
\end{remark}

\subsection{Spectral Gap Filling} \label{SectionGapFilling}

Let $X$ be a complete Riemannian manifold and let $H$ be a Hamiltonian. We take $H$ to be of Laplace type, meaning that $H - \Delta$ is an operator of order at most one, where $\Delta$ is the Laplace--Beltrami operator of $X$. We assume that $H$ is symmetric, non-negative, and essentially self-adjoint on the domain $C_c^\infty(X) \subset L^2(X)$, so that there exists a unique self-adjoint extension (still non-negative), which we again denote by $H$. One then has the following lemma \cite[Prop.~3.6]{Roe-coarse-book}.

\begin{lemma} \label{PropFunctionsInCStarX}
 For any $\varphi \in C_0(\RR)$, the operator $\varphi(H)$ defined by functional calculus is contained in the Roe algebra $C^*(X)$. 
\end{lemma}

In particular, if $S \subset \mathrm{spec}(H)$ is a compact part of the spectrum that is separated from the rest of the spectrum by spectral gaps, then the spectral projection $P_S$ is contained in $C^*(X)$, as there exists a continuous function $\varphi$ that is equal to one on $S$ and zero on $\mathrm{spec}(H) \setminus S$, and $P_S = \varphi(H)$ for such a function (see Fig.\ \ref{fig:spectral.projection}). We obtain an associated class
\begin{equation}
  [P_S] \in K_0(C^*(X))
\end{equation}
in the $K$-theory of the (bulk) Roe algebra.

\medskip

Let now $W \subseteq X$ be a regular closed subset with interior $U = \mathring{W}$. We can then consider the operator $H$ on the space $C^\infty_c(U)$ of smooth functions with compact support not touching the boundary; this gives a symmetric, densely defined, unbounded operator on $L^2(W)$. 
Let $H_W$ be the self-adjoint extension obtained by taking either Dirichlet boundary conditions or, in case that $W$ is sufficiently regular to define the normal derivative almost everywhere at the boundary, Neumann boundary conditions. The following result can be found in \cite{LTcobordism}, Thm.~3.2.

\begin{lemma} \label{LemmaDifference}
For each $\varphi \in C_0(\RR)$, we have $\varphi(H_W) \in C^*(W)$. Moreover,  the difference $\varphi(H_W) - r(\varphi(H))$ is contained in $C^*_W(\partial W)$.
\end{lemma}

\begin{proof}
We repeat the proof of Lemmas~\ref{PropFunctionsInCStarX} and \ref{LemmaDifference} for convenience of the reader. It relies on the cosine transform identity
\begin{equation} \label{FourierTransformFormula}
  \psi(D) = \frac{1}{\pi} \int_0^\infty \hat{\psi}(t) \cos(t D) \mathrm{d} t,
\end{equation}
for even functions $\psi$, which holds abstractly for any self-adjoint operator $D$ on a Hilbert space $\mathcal{H}$. Now when taking $D = \sqrt{H}$ or $\sqrt{H_W}$, it is well-known (see, e.g., \cite[Chapter~6,~Prop.~1.3]{Taylor} for the case with boundary) that the operators $\cos(t D)$ have propagation speed at most $|t|$ for any $t \in \RR$. Hence if given $\varphi \in C_0(\RR)$, one sets $\psi(t) = \varphi(t^2)$ and assumes in addition that $\hat{\psi}$ is compactly supported in $[-R, R]$, it follows directly from \eqref{FourierTransformFormula} that $\varphi(H)$ has propagation speed at most $R$. Because functions satisfying this support assumption are dense in $C_0(\RR)$, for general $\varphi$, the operator $\varphi(H)$ is a norm limit of finite propagation operators. Since local compactness is a consequence of ellipticity  (as shown in \cite[Prop.~10.5.1]{HigsonRoeBook}), this shows that $\varphi(H)$ is in fact contained in the Roe algebra; this argument works for both $H$ and $H_W$.

That indeed $\varphi(H_W) - r(\varphi(H)) \in C^*_W(\partial W)$ can be seen as follows. To begin with, it is again an abstract observation for any self-adjoint operator $D$ on a Hilbert space $\mathcal{H}$ that for any $u \in \mathcal{H}$, the $\mathcal{H}$-valued function $u_t = \cos(tD)u$ satisfies the \emph{abstract wave equation} with initial conditions
\begin{equation} \label{WaveEquation}
  \ddot{u}_t + D^2 u_t = 0, \qquad u_0 = u, \quad \dot{u}_0 = 0.
\end{equation}
For $u \in C^\infty_c(U) \subset C^\infty_c(X)$, denote by $u_t^{(X)} \in L^2(X)$ and $u_t^{(W)} u \in L^2(W)$ the two solutions obtained from setting $D=\sqrt{H}$, respectively $D=\sqrt{H_W}$. Since the corresponding wave operators have finite propagation speed, both $u_t^{(X)}$ and $u_t^{(W)}$ are contained in the $|t|$-ball around $\mathrm{supp}(u))$. Now because $H$ and $H_W$ coincide on smooth, compactly supported functions and since solutions to the wave equation are unique, it follows that $u_t^{(X)}|_W = u_t^{(W)}$ whenever this ball is contained in $U$. With a view in \eqref{FourierTransformFormula}, this implies again that if $\psi(t) = \varphi(t^2)$ has compactly supported Fourier transform, then $\varphi(H_W) - r(\varphi(H))$ is supported near the boundary; in general, the difference $\varphi(H_W) - r(\varphi(H))$ is a norm limit of operators supported near the boundary, hence contained in the localized Roe algebra $C^*_W(\partial W)$.
\end{proof}

\begin{remark} \label{RemarkGeneralBC}
It is clear from the proof that it is not important to restrict to Dirichlet or Neumann boundary conditions in Lemma~\ref{LemmaDifference}. Going through the proof, it is easy to see that the result is true for any self-adjoint extension $H_W$ such that (1) $H_W$ is still non-negative; (2) the wave operators $\cos(t \sqrt{H_W})$ have finite propagation speed; (3) the operator is \emph{elliptic} in the sense that the inclusion operator of $\mathrm{dom}(H_W)$ (as a Banach space with the graph norm) into $L^2(W)$ is compact.
\end{remark}

\begin{figure}[h]
\begin{center}
\begin{tikzpicture}
%\draw[help lines] (0,0) grid (6,6);
       \begin{axis}[axis lines=none,xscale=1.1,yscale=0.6]
         \addplot[samples=20, smooth, domain=0.5:1]
           plot (\x, {0.5*(1+ sin(deg(2*pi*(\x-0.75))))});
           \addplot[samples=20, smooth, domain=2:2.5]
           plot (\x, {0.5*(1- sin(deg(2*pi*(\x-2.25))))});
            \addplot[samples=3, smooth, domain=-0.5:0.5]
           plot (\x, {0});
            \addplot[samples=3, smooth, domain=1:2]
           plot (\x, {1});
           \addplot[samples=3, smooth, domain=2.5:3.5]
           plot (\x, {0});
      
       \end{axis}
            \draw[line width=6pt] (0.5,-0.3) -- (1.5,-0.3);
             \draw[line width=6pt] (3,-0.3) -- (3.5,-0.3);
             \draw[line width=6pt] (3.7,-0.3) -- (4.5,-0.3);
              \draw[|-|] (3,-1) -- (3.5,-1);
              \node at (3.8,-1.6) {$S\subset [b,c]$};
              \draw[|-|] (3.7,-1) -- (4.5,-1);
              \draw[line width=6pt] (6,-0.3) -- (7,-0.3);
              \node at (2.2,1.5) {$\varphi$}; 
               \node[below] at (1.5,-0.4) {$a$}; 
                \node[below] at (3,-0.3) {$b$}; 
                 \node[below] at (4.5,-0.4) {$c$}; 
                  \node[below] at (6,-0.3) {$d$}; 
                \node at (3,3.8) {};
                \draw[line width=1pt] (3.7,-0.3) -- (6,-0.3);
           \end{tikzpicture}
  \end{center}
\caption{Thick horizontal lines indicate the spectrum of $H$ as a subset of $\RR$. For $S$ a compact separated part of spec($H$), a function $\varphi\in C_0(\RR)$ which equals $1$ on $S$ and 0 on the rest of spec($H$), is plotted. For such a $\varphi$, the spectral projection $P_S$ for $S$ equals $\varphi(H)$. The half-space operator $H_W$ may acquire spectra in the gaps of spec($H$). If the coarse index of $P_S$ is non-trivial, then $\varphi(H_W)$ is never a projection, whatever $\varphi$ we choose. Thus $\mathrm{spec}(H_W)$ must completely fill up either $(a,b)$ or $(c,d)$, as indicated by the thin horizontal line bridging a gap of spec($H$). }\label{fig:spectral.projection}
\end{figure}
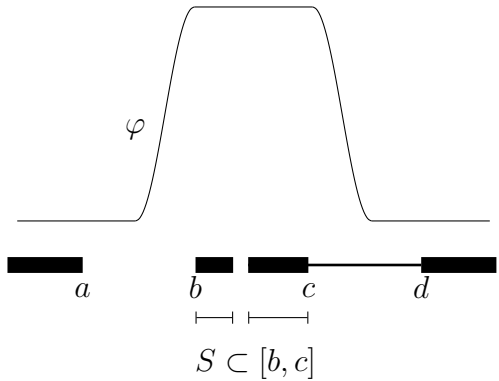

Let now $S \subset \mathrm{spec}(H)$ be a compact part of the spectrum, separated from the rest of the spectrum by gaps. This leads to a spectral projection $P_S \in C^*(X)$, as discussed above. Explicitly, let $S \subset [b, c]=[{\rm inf}(S),{\rm sup}(S)]$, so that $(a, b)$ and $(c, d)$ are open subsets of the resolvent set of $H$, for $a < b \leq c < d$. We then have the following result.

\begin{theorem}[Gap-Filling]\label{thm:gap.filling}
If $\partial_0([P_S]) \neq 0$ in $K_1(C^*_W(\partial W))$, then one of the spectral gaps $(a, b)$ or $(c, d)$ of $H$ adjacent to $S$, is completely contained in $\mathrm{spec}(H_W)$.
\end{theorem}

\begin{proof}
Suppose that neither $(a, b)$ nor $(c, d)$ are contained in $\mathrm{spec}(H_W)$. Since the spectrum is closed, there exist non-empty open subintervals $(a^\prime, b^\prime) \subset (a, b)$ and $(c^\prime, d^\prime) \subset (c, d)$ that are completely contained in the resolvent set of $H_W$. We can therefore choose a continuous function $\varphi$ that is constant equal to one on $[b^\prime, c^\prime]$ and zero on $(-\infty, a^\prime]$ and $[d^\prime, \infty)$. For such a function $\varphi$, we have $P_S = \varphi(H)$ (as discussed above), and by Lemma~\ref{LemmaDifference}, $\varphi(H_W)$ is a self-adjoint lift of the element $\tilde{r}(\varphi(H)) \in C^*(W) / C^*_W(\partial W)$. However, since $\varphi$ only takes the values $0$ and $1$ on $\mathrm{spec}(H_W)$, $\varphi(H_W)$ is a projection, and hence
\begin{equation*}
  \partial_0([P_S]) = \delta_0 (\tilde{r}_*[\varphi(H)]) = [\exp(2 \pi i \varphi(H_W))] = [1],
\end{equation*}
which is the trivial class.
Here we used that the $K$-theory boundary map $\delta_0$ has the explicit description as an exponential map: For general projections $p$ in $C^*(W)/C^*_W(\partial W)$, one has $\delta_0([p]) = \exp(2 \pi i \tilde{p})$, where $\tilde{p} \in C^*(W)$ is any self-adjoint lift of $p$.
\end{proof}

\begin{remark}\label{rem:equivariant.advantage}
If $H$ is equivariant with respect to a cocompact, and possibly projective, action of a discrete group $\Gamma$, and the distance function of $\partial W$ is unbounded, the same gap-filling result was similarly obtained (Theorem 3.4 of \cite{LTcobordism}) by analysing the exponential map for the quasi-invariant sequence Eq.\ \eqref{ShortExactQuasi}. For Landau Hamiltonians $H_\theta$ studied later on, we may pick $\Gamma=\ZZ^2$ (Euclidean case) or $\Gamma$ a surface group (hyperbolic case) acting on $X$. Then $H_\theta$ is invariant under the projective action of $\Gamma$ on $L^2(X)$ by \emph{magnetic translations} \cite{CHMM, MThyperbolic}.

Although the equivariant approach is somewhat more complicated, there are two advantages. First, non-vanishing of the equivariant coarse index of $P_S$ may be easier to verify for certain Hamiltonians, e.g., \emph{Chern insulators}, see \S5.2 of \cite{LTcobordism}. Second, in the equivariant setting, it can be shown that $\mathrm{spec}(H) \subseteq \mathrm{spec}(H_W)$ (Corollary 3.3 of \cite{LTcobordism}). In contrast, when $H$ is not $\Gamma$-invariant, Theorem \ref{thm:gap.filling} still applies, but $\mathrm{spec}(H) \subseteq \mathrm{spec}(H_W)$ is false in general (i.e., gaps may be \emph{introduced} into spec($H$)). 
\end{remark}

\section{Landau levels as kernels of twisted Dirac operators}\label{sec:Landau.Dirac}

\subsection{Magnetic Laplacians and Dirac operators}\label{sec:Lich}

Let $X$ be a contractible oriented two-dimensional Riemannian manifold and let $\omega$ be its volume form. By the Poincar\'{e} Lemma, for any given smooth real-valued function $\theta$ on $X$, we can find a 1-form $A_\theta \in \Omega^1(X)$ such that $d A_\theta = \theta \cdot \omega$. The {\em magnetic Laplacian} is then the operator
\begin{equation} \label{Eq:LandauHamiltonian}
H_\theta:=(d-iA_\theta)^*(d-iA_\theta), 
\end{equation}
acting on complex-valued functions on $X$.

From a geometric point of view, given a function $\theta$, there exists a Hermitian line bundle $\mathcal{L}_\theta$ with a connection $\nabla_\theta$ having curvature $F_\theta = -i \theta \cdot \omega$; again since $X$ is contractible, this line bundle must be globally trivial, and it is unique up to metric and connection preserving isomorphism (gauge transformations). Under a global trivialization, the connection $\nabla_\theta$ of $\mathcal{L}_\theta$ is sent to the operator $d - i A_\theta$ for some $A_\theta \in \Omega^1(X)$ with $dA_\theta = \theta \cdot \omega$, so that the magnetic Laplacian, Eq.\ \eqref{Eq:LandauHamiltonian}, is identified with the connection Laplacian $\nabla_\theta^* \nabla_\theta$ of $\mathcal{L}_\theta$ under this trivialization. We have $\mathcal{L}_{\theta_1} \otimes \mathcal{L}_{\theta_2} \cong \mathcal{L}_{\theta_1 + \theta_2}$ as Hermitean line bundles with connection.

\medskip

Let $\mathcal{S}$ be the spinor bundle over $X$. Since we are in two dimensions, its typical fiber is $\CC^2$, and decomposing with respect to the $(\pm1)$-eigenbundles of the grading operator $\sigma_3 := i c(e_1)c(e_2)$, it splits as the direct sum $\mathcal{S} = \mathcal{S}^+ \oplus \mathcal{S}^-$ of two Hermitian line bundles (the {\em positive} and {\em negative chirality spinors}). Here $e_1$, $e_2$ denotes a local orthonormal frame of the tangent bundle of $X$ and $c(v)$ denotes Clifford multiplication in $\mathcal{S}$ by the vector $v$. A standard computation (using, e.g., Prop.~3.43 of \cite{BerlineGetzlerVergne}) shows that  the curvature of $\mathcal{S}$ is given by $F^{\mathcal{S}}(e_1, e_2) =i\frac{R}{4} \sigma_3$, where $R$ is the scalar curvature. Therefore, using that in two dimensions, the curvature is the only invariant of a Hermitean line bundle with connection, we can identify $\mathcal{S}^\pm \cong \mathcal{L}_{\mp \frac{R}{4}}$. 

\medskip

We can now form the Dirac operator $\slashed{D}$,
acting on $\mathcal{S}$, and its twisted version $\slashed{D}_\theta$, acting on $\mathcal{S} \otimes \mathcal{L}_\theta$. The Lichnerowicz--Schr\"odinger--Weitzenb\"ock formula (see, e.g., Prop.~3.52 of \cite{BerlineGetzlerVergne}) then states that its square is related to the connection Laplacian by the formula
\begin{equation} \label{eqn:Lichnerowicz}
\slashed{D}_\theta^2= \nabla_{\mathcal{S}\otimes\mathcal{L}_\theta}^*\nabla_{\mathcal{S}\otimes\mathcal{L}_\theta}+\frac{R}{4}-i\theta \cdot c(e_1)c(e_2).
\end{equation}
With respect to the splitting $\mathcal{S} \otimes \mathcal{L}_\theta = (\mathcal{S}^+ \otimes \mathcal{L}_\theta) \oplus (\mathcal{S}^- \otimes \mathcal{L}_\theta)$ of the spinor bundle into its positive and negative chirality part (i.e., the eigenbundles of $\sigma_3$), the twisted Dirac operator and its square take the form
\begin{equation} \label{eqn:DiracMatrix}
\slashed{D}_\theta = \begin{pmatrix} 0 & \slashed{D}^-_\theta \\ \slashed{D}^+_\theta & 0 \end{pmatrix} \qquad \text{and} \qquad
\slashed{D}_\theta^2 = \begin{pmatrix} \slashed{D}^-_\theta \slashed{D}_\theta^+ & 0 \\ 0 & \slashed{D}^+_\theta \slashed{D}_\theta^- \end{pmatrix}.
\end{equation}
By the observations above, we have 
\begin{equation}
\mathcal{S}^\pm \otimes \mathcal{L}_\theta \cong \mathcal{L}_{\theta\mp \frac{R}{4}},
\end{equation}
hence the connection Laplacian in \eqref{eqn:Lichnerowicz} can be identified with a direct sum of Landau Hamiltonians corresponding to the parameters $\theta \mp \frac{R}{4}$. Therefore, remembering that $\sigma_3 = i c(e_1) c(e_2)$ and replacing $\theta$ by $\theta+\frac{R}{4}$, we obtain the following result.

\begin{proposition}\label{prop:Landau.Dirac}
For any $\theta\in C^\infty(X)$, the magnetic Laplacian and twisted Dirac operators are related by the formula
\begin{equation}\label{eqn:Lich}
  \slashed{D}_{\theta+\frac{R}{4}}^2 = \begin{pmatrix} H_{\theta} -\theta & 0 \\ 0 & H_{\theta+\frac{R}{2}} +\theta +\frac{R}{2} \end{pmatrix}
\end{equation}
 with respect to the splitting into positive and negative chirality spinors. 
\end{proposition}

\begin{remark}[Spectral supersymmetry]\label{rem:spectral.susy}
Prop.\ \ref{prop:Landau.Dirac} provides the geometric origin for the statement that $H_{\theta}-\theta$ has \emph{supersymmetric partner} $H_{\theta+\frac{R}{2}}+\theta + \frac{R}{2}$. Consequently, the two operators must share the same non-zero spectrum, see \S 5 of \cite{Thaller}, see also \cite{Moller}. Here, there is a standard way to make these operators self-adjoint when $X$ is complete. Furthermore, the latter operator has the same form as the former operator, except for a parameter shift $\theta\mapsto\theta+\frac{R}{2}$ and an extra parameter-dependent scalar $2\theta+\frac{R}{2}$. In case $\theta$ and $R$ are constant functions, such a relationship between supersymmetric partners is sometimes called the \emph{shape invariance} property, and it allows for a ``ladder operator'' method for spectral computation \cite{Gendenshtein, Benedict}. For $X$ the hyperbolic plane, a generalisation of this technique was used in \cite{Inahama2} to study the essential spectral properties of $H_\theta$ with asymptotically constant $\theta$.
\end{remark}

\subsection{Landau levels and Dirac operators for the hyperbolic plane}\label{sec:raising.lowering}
In this section, we consider the the hyperbolic plane $X = \mathbb{H}$, and constant $\theta \in \RR \setminus \{0\}$. In this case, the magnetic Laplacian $H_\theta$, Eq.\ \eqref{Eq:LandauHamiltonian}, is called the \emph{Landau Hamiltonian}. Both $H_\theta$ and the twisted Dirac operators $\slashed{D}_\theta$, are essentially self-adjoint on compactly supported smooth functions (respectively compactly supported smooth spinors), and their unique self-adjoint extensions to unbounded operators on $L^2(X)$, respectively $L^2(X, \mathcal{S})$ are again denoted by the same symbols. We will exploit the spectral supersymmetry (Remark \ref{rem:spectral.susy}) to study the spectrum of $H_\theta$, as illustrated in Fig.\ \ref{fig:hyp.raising.lowering}. While similar methods were used in \cite{Inahama2} to compute the spectrum, our presentation stresses the geometric relationship with Dirac operators so that (coarse) index theory can be applied later in \S\ref{sec:coarse.Landau}. 

\vspace{0.5em}
Since the scalar curvature on $\mathbb{H}$ is constant, $R = -2$, we obtain from Prop.~\ref{prop:Landau.Dirac} the fundamental relationships, which hold for all $\theta\in\RR$, 
\begin{align}
\slashed{D}_{\theta-\frac{1}{2}}^2 &= \begin{pmatrix}\slashed{D}_{\theta-\frac{1}{2}}^-\slashed{D}_{\theta-\frac{1}{2}}^+ & 0 \\ 0 & \slashed{D}_{\theta-\frac{1}{2}}^+\slashed{D}_{\theta-\frac{1}{2}}^-\end{pmatrix}=\begin{pmatrix} H_\theta -\theta & 0 \\ 0 & H_{\theta-1} +\theta -1 \end{pmatrix}\geq 0, \label{eqn:Dirac.squared}\\
\slashed{D}_{\theta+\frac{1}{2}}^2 &= \begin{pmatrix}\slashed{D}_{\theta+\frac{1}{2}}^-\slashed{D}_{\theta+\frac{1}{2}}^+ & 0 \\ 0 & \slashed{D}_{\theta+\frac{1}{2}}^+\slashed{D}_{\theta+\frac{1}{2}}^-\end{pmatrix}=\begin{pmatrix} H_{\theta+1} -\theta-1 & 0 \\ 0 & H_{\theta} +\theta \end{pmatrix}\geq 0,\label{eqn:Dirac.squared2}
\end{align}
where the second follows from the first upon replacing $\theta$ by $\theta-1$. Inspecting the top-left piece of Eq.\ \eqref{eqn:Dirac.squared} if $\theta\geq 0$ (resp.\ bottom-right piece of Eq.\ \eqref{eqn:Dirac.squared2} if $\theta\leq 0$), we obtain an easy lower bound $H_\theta\geq |\theta|$. 

The value $|\theta|$ is the \emph{lowest Landau level}, but it is only attained in the spectrum of $H_\theta$ when $|\theta|\geq\frac{1}{2}$. More generally, the spectrum of $H_\theta$ ($\theta \neq 0$) consists of isolated eigenvalues, called \emph{Landau levels},
\begin{equation}
\lambda_{m, \theta} =(2m+1)|\theta |-m(m+1),\quad m=0,1,\ldots m_{\rm max}< |\theta |-\frac{1}{2},\label{eqn:hyperbolic.Landau.levels2}
\end{equation}
as well as a continuous part $[\frac{1}{4} + \theta^2, \infty)$ above the Landau levels, see \cite {Comtet, Inahama2}. 
When $|\theta|>m+\frac{1}{2}$, the $m$-th Landau level $\lambda_{m,\theta}$ is isolated, and we denote by $\mathcal{E}_{m, \theta}$ its corresponding \emph{$m$-th Landau eigenspace}.

\begin{lemma}\label{lem:lowest.Landau.index}
For $|\theta| > \frac{1}{2}$, the eigenspace to the lowest Landau level $\lambda_{0,\theta}=|\theta|$ is 
\begin{equation*}
\mathcal{E}_{0,\theta} = \begin{cases} \ker(\slashed{D}^+_{\theta-\frac{1}{2}}) & \text{if}~\theta >\frac{1}{2},\\
\ker(\slashed{D}^-_{\theta+\frac{1}{2}}) & \text{if}~\theta <-\frac{1}{2}.
\end{cases}
\end{equation*}
On the other hand, if $\theta > \frac{1}{2}$, then $\ker(\slashed{D}^-_{\theta-\frac{1}{2}}) = 0$ and if $\theta < -\frac{1}{2}$, then $\ker(\slashed{D}^+_{\theta+\frac{1}{2}}) = 0$.
\end{lemma}

\begin{proof}
Let first $\theta >\frac{1}{2}$. From %\eqref{eqn:DiracMatrix} and 
\eqref{eqn:Dirac.squared}, we have 
\begin{equation*}
 \slashed{D}^-_{\theta-\frac{1}{2}} \slashed{D}^+_{\theta - \frac{1}{2}} = H_\theta - \theta \qquad \text{and} \qquad  \slashed{D}^+_{\theta-\frac{1}{2}} \slashed{D}^-_{\theta - \frac{1}{2}} = H_{\theta-1} + \theta - 1.
\end{equation*}  
We claim that $H_{\theta-1} + \theta-1$ is strictly positive, so that ${\rm ker}(\slashed{D}_{\theta - \frac{1}{2}}^-)$ is trivial. In the case $\theta >1$, this is automatic as $H_{\theta-1}$ is positive. For $\theta\in (\frac{1}{2},1]$, $H_{\theta-1}$ has no isolated Landau levels so that $H_{\theta-1}\geq \frac{1}{4}+(\theta-1)^2$, and therefore $H_{\theta-1}+\theta-1\geq (\theta-\frac{1}{2})^2 > 0$ as claimed. It follows that
\begin{equation*}
{\rm ker}(\slashed{D}^+_{\theta-\frac{1}{2}})={\rm ker}(\slashed{D}^-_{\theta-\frac{1}{2}}\slashed{D}^+_{\theta-\frac{1}{2}})={\rm ker}(H_\theta-\theta)\equiv\mathcal{E}_{0,\theta}.
\end{equation*}
For $\theta < -\frac{1}{2}$, a similar argument shows that the top left piece of \eqref{eqn:Dirac.squared2} is strictly positive, so ${\rm ker}(\slashed{D}^+_{\theta+\frac{1}{2}})=0$, whence it follows that 
\begin{equation*}
{\rm ker}(\slashed{D}^-_{\theta+\frac{1}{2}})={\rm ker}(\slashed{D}^+_{\theta+\frac{1}{2}}\slashed{D}^-_{\theta+\frac{1}{2}})={\rm ker}(H_\theta+\theta)\equiv\mathcal{E}_{0,\theta}.
\end{equation*}
\end{proof}

\begin{lemma}\label{lem:isolated.zero}
$0$ is isolated in the spectrum of $\slashed{D}_a$ whenever $0\neq a\in\RR$.
\end{lemma}
\begin{proof}
From Eq.\ \eqref{eqn:Dirac.squared} and \eqref{eqn:Dirac.squared2}, we may reexpress
\begin{align*}
{\rm ker}(\slashed{D}_a)={\rm ker}(\slashed{D}_a^2)&={\rm ker}(\slashed{D}_a^+)\oplus {\rm ker}(\slashed{D}_a^-)\\
&={\rm ker}\bigl(H_{a+\frac{1}{2}}-(a+\tfrac{1}{2})\bigr)\oplus {\rm ker}\bigl(H_{a-\frac{1}{2}}+(a-\tfrac{1}{2})\bigr).
\end{align*}
Suppose $a>0$ (a similar argument takes care of the $a<0$ case). Inspecting the spectrum of $H_{a+\frac{1}{2}}$ (Eq.\ \eqref{eqn:hyperbolic.Landau.levels2}) shows that 0  is isolated in the spectrum of $H_{a+\frac{1}{2}}-(a+\frac{1}{2})$. For the second piece of ${\rm ker}(\slashed{D}_a)$, we have $H_{a-\frac{1}{2}}+(a-\tfrac{1}{2})$ strictly positive if $a>\tfrac{1}{2}$, while for $0<a\leq\tfrac{1}{2}$, we have $0\leq |a-\frac{1}{2}|<\tfrac{1}{2}$, and thus $H_{a-\frac{1}{2}}+(a-\frac{1}{2})\geq \left(\frac{1}{4}+(a-\frac{1}{2})^2\right)+(a-\tfrac{1}{2})=(\frac{1}{2}+(a-\tfrac{1}{2}))^2=a^2>0$ is again strictly positive. Overall, $0$ is isolated in ${\rm spec}(\slashed{D}_a)$.
\end{proof}

\begin{proposition}\label{prop:induction}
For $0\neq a\in\RR$, define the operator $V_a$ by
\begin{equation} \label{DefinitionV}
  V_a = \begin{cases} (\slashed{D}^+_a \slashed{D}_a^-)^{-1/2} \slashed{D}^+_a = \slashed{D}^+_a (\slashed{D}^-_a \slashed{D}_a^+)^{-1/2} & \text{on} ~\ker(\slashed{D}^+_a)^\perp \\
  0 &\text{on}~\ker(\slashed{D}^+_a).
  \end{cases}
\end{equation}
Then for any $1\leq m < |\theta| - \frac{1}{2}$, restriction provides unitary isomorphisms
\begin{align*}
  V_{\theta-\frac{1}{2}}&: \mathcal{E}_{m, \theta} \longrightarrow \mathcal{E}_{m-1, \theta-1} & &(\theta >m+\tfrac{1}{2})\\
  V_{\theta+\frac{1}{2}}^*&: \mathcal{E}_{m, \theta} \longrightarrow \mathcal{E}_{m-1, \theta+1} & &(\theta < -m-\tfrac{1}{2}).
\end{align*} 
\end{proposition}

\begin{proof}
Suppose $\theta>m+\frac{1}{2}$, and write $a=\theta-\frac{1}{2}$. Let $\xi \in \mathcal{E}_{m, \theta}$, so that 
\begin{equation*}
\slashed{D}^-_a\slashed{D}^+_a\xi\equiv\slashed{D}^-_{\theta - \frac{1}{2}}\slashed{D}^+_{\theta - \frac{1}{2}} \xi = (H_\theta - \theta)\xi = (\lambda_{m, \theta} - \theta)\xi.
\end{equation*}
Applying $V_a$ to both sides, and using \eqref{eqn:Dirac.squared} again, we find
\begin{align*}
  (\lambda_{m, \theta} - \theta) V_{a}\xi &= \slashed{D}^+_{a}(\slashed{D}^-_{a} \slashed{D}_{a}^+)^{-1/2}\slashed{D}^-_{a}\slashed{D}^+_{a} \xi \\
  &= \slashed{D}^+_{a}\slashed{D}^-_{a}\slashed{D}^+_{a} (\slashed{D}^-_{a} \slashed{D}_{a}^+)^{-1/2}\xi \\
  &=(\slashed{D}^+_{\theta-\frac{1}{2}}\slashed{D}^-_{\theta-\frac{1}{2}})V_a \xi= (H_{\theta-1} + \theta - 1) V_{a} \xi.
\end{align*}
Rearranging, and using $\lambda_{m, \theta} - 2\theta +1 =  \lambda_{m-1, \theta-1}$, we obtain
\begin{equation*}
H_{\theta - 1}( V_{a} \xi) = \lambda_{m-1,\theta-1} (V_{a} \xi)\quad\Rightarrow\quad V_{a} \xi \in \mathcal{E}_{m-1, \theta-1}. 
\end{equation*}
In a similar way, for $\zeta\in\mathcal{E}_{m-1,\theta-1}$, we have 
\begin{align*}
(\lambda_{m-1,\theta-1}+\theta-1)V_a^*\zeta&= V_a^*(H_{\theta-1}+\theta-1)\zeta \\
&=(\slashed{D}^-_{a}\slashed{D}^+_{a})^{-1/2}\slashed{D}^-_{a}(\slashed{D}^+_{a}\slashed{D}^-_{a})\zeta\\
&=(\slashed{D}^-_{a}\slashed{D}^+_{a})(\slashed{D}^-_{a}\slashed{D}^+_{a})^{-1/2}\slashed{D}^-_{a}\zeta=(H_\theta-\theta)V_a^*\zeta,
\end{align*}
so $H_\theta(V_a^*\zeta)=(\lambda_{m-1,\theta-1}+2\theta-1)(V_a^*\zeta)=\lambda_{m,\theta}(V_a^*\zeta)$, i.e., $V_a^*\zeta\in\mathcal{E}_{m,\theta}$. It is easy to see that $V_a:\mathcal{E}_{m,\theta}\rightarrow\mathcal{E}_{m-1,\theta-1}$ has inverse map the ``raising'' operator $V_a^*$ (restricted to $\mathcal{E}_{m-1,\theta-1}$).

Now suppose $\theta<-m-\frac{1}{2}$, and $\xi \in \mathcal{E}_{m, \theta}$. This time, let $a=\theta+\frac{1}{2}$, and repeat the above argument with $V^*_{a}$ (with an eye on \eqref{eqn:Dirac.squared2}). We obtain 
\begin{equation*}
(\lambda_{m,\theta}+\theta)V^*_{a}\xi=(H_{\theta+1}-\theta-1)V^*_{a}\xi. 
\end{equation*}
Since $\lambda_{m-1,\theta+1}=\lambda_{m,\theta}+2\theta+1$, we obtain $V^*_{a}\xi\in\mathcal{E}_{m-1,\theta+1}$. Similarly, for $\zeta\in\mathcal{E}_{m-1,\theta+1}$, we have $V_a\zeta\in\mathcal{E}_{m,\theta}$. It is again easy to see that $V_a$ is the inverse to $V_a^*:\mathcal{E}_{m,\theta}\rightarrow\mathcal{E}_{m-1,\theta+1}$.
\end{proof}

\begin{remark}\label{rem:lowering.remark}
Note that via the spectral theorem, 
\begin{equation*}
{\rm sgn}(\slashed{D}_a)=\begin{pmatrix} 0 & V^*_a \\ V_a & 0\end{pmatrix}.
\end{equation*}
For positive $\theta$, the lowering and raising of Landau levels by $V_{\theta-\frac{1}{2}}$ and $V_{\theta-\frac{1}{2}}^*$ respectively, is illustrated in Fig.\ \ref{fig:hyp.raising.lowering}. 
\end{remark}

\begin{corollary}\label{cor:Landau.Dirac.kernel}
The eigenspace for any isolated Landau level of $H_\theta$ is unitarily isomorphic to the kernel of a twisted Dirac operator. Specifically,
\begin{equation*}
  \mathcal{E}_{m, \theta} \cong \begin{cases} \ker( \slashed{D}^+_{\theta - m-\frac{1}{2}}), & \text{if}~ \theta > m+\frac{1}{2}, \\ \ker( \slashed{D}^-_{\theta + m+\frac{1}{2}}), & \text{if} ~\theta <-m-\frac{1}{2}. \end{cases}
\end{equation*}
\end{corollary}
\begin{proof}
For the $m$-th isolated Landau level, iterating Prop.\ \ref{prop:induction} gives isomorphisms lowering the Landau levels,
\begin{align*}
\mathcal{E}_{m,\theta}&\cong\mathcal{E}_{m-1,\theta-1}\cong\mathcal{E}_{m-2,\theta-2}\cong\ldots \cong\mathcal{E}_{0,\theta-m}, & (\theta&>m+\tfrac{1}{2}),\\
\mathcal{E}_{m,\theta} &\cong \mathcal{E}_{m-1,\theta+1} \cong \mathcal{E}_{m-2,\theta+2}\cong\ldots \cong \mathcal{E}_{0,\theta+m}, & (\theta&<-m-\tfrac{1}{2}).
\end{align*}
Then Lemma \ref{lem:lowest.Landau.index} identifies the last subspace as the kernel of a Dirac operator.
\end{proof}

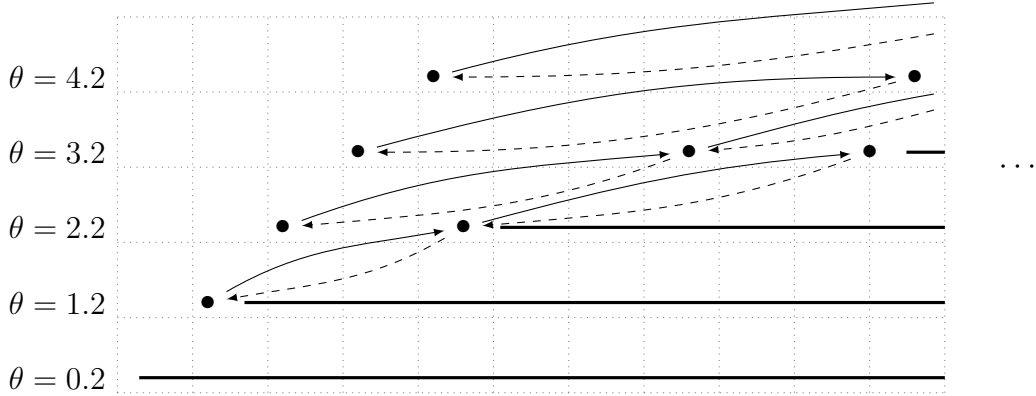
\begin{figure}
\begin{center}
\begin{tikzpicture}
\draw[step=1cm,gray, thin, dotted] (0,0) grid (11,5);
\draw[very thick] (0.29,0.2)--(11,0.2);
%\filldraw (1.2,1.2) circle (2pt);
\draw[very thick] (1.69,1.2)--(11,1.2);
\draw[very thick] (5.09,2.2)--(11,2.2);
\draw[very thick] (10.49,3.2)--(11,3.2);
\draw (-0.8,0.2) node {$\theta=0.2$};
\draw (-0.8,1.2) node {$\theta=1.2$};
\draw (-0.8,2.2) node {$\theta=2.2$};
\draw (-0.8,3.2) node {$\theta=3.2$};
\draw (-0.8,4.2) node {$\theta=4.2$};
\node (1-0) at (1.2,1.2) {$\bullet$};
\node (2-0) at (2.2,2.2) {$\bullet$};
\node (2-1) at (4.6,2.2) {$\bullet$};
\node (3-0) at (3.2,3.2) {$\bullet$};
\node (3-1) at (7.6,3.2) {$\bullet$};
\node (3-2) at (10,3.2) {$\bullet$};
\node (4-0) at (4.2,4.2) {$\bullet$};
\node (4-1) at (10.6,4.2) {$\bullet$};
\node (5-1) at (11,4.8) {};
\node (5-1prime) at (11,5.2) {};
\node (4-2) at (11,3.8) {};
\node (4-2prime) at (11,4.0) {};
\node at (12,3) {$\cdots$};
%\draw[-latex] (2-1) to[out=-150,in=10] node[midway,font=\scriptsize,below] {$U_{1,2.2}$} (1-0);
\draw[-latex,dashed] (2-1) to[out=-150,in=10] (1-0);
\draw[-latex,dashed] (3-1) to[out=-160,in=5] (2-0);
\draw[-latex,dashed] (3-2) to[out=-160,in=5] (2-1);
\draw[-latex,dashed] (4-1) to[out=-165,in=0] (3-0);
\draw[-latex,dashed] (5-1) to[out=-170,in=0] (4-0);
\draw[-latex,dashed] (4-2) to[out=-165,in=5] (3-1);
\draw[-latex] (1-0) to[out=30,in=-170] (2-1);
\draw[-latex] (2-0) to[out=20,in=-175] (3-1);
\draw[-latex] (2-1) to[out=15,in=-175] (3-2);
\draw[-latex] (3-0) to[out=15,in=-180] (4-1);
\draw (4-0) to[out=15,in=-175] (5-1prime);
\draw (3-1) to[out=15,in=-170] (4-2prime);

%\draw[-latex,dashed] (4-2) to[out=-160,in=5] (3-2);
\end{tikzpicture}
\end{center}
\caption{Spectrum of hyperbolic Landau Hamiltonian $H_\theta$ for some values of $\theta$. The continuous spectrum is shown as horizontal lines, while $\bullet$ labels the isolated Landau levels. The $m$-th Landau level for $H_\theta$ is lowered to the $(m-1)$-th Landau level for $H_{\theta-1}$ by $V_{\theta-\frac{1}{2}}$ (dashed arrow), except when $m=0$, i.e., a lowest Landau level. Similarly, $V_{\theta+\frac{1}{2}}^*$ raises the Landau levels (solid arrows).}\label{fig:hyp.raising.lowering}
\end{figure}

\begin{remark}
That the half-infinite interval $[\frac{1}{4}+\theta^2,\infty)$ lies in the spectrum of $H_\theta$, may be shown, e.g.\ by constructing a Weyl sequence (Lemma 4.3 of \cite{Inahama2}). In the above calculations, we also used prior knowledge of the existence of isolated Landau levels $\lambda_{m,\theta}$ lying below $\frac{1}{4}+\theta^2$, to justify their study in the first place. Interestingly, it is actually possible to show algebraically, by similar bootstrap methods, that the values $\lambda_{m,\theta}$ must be isolated if they occur in the spectrum of $H_\theta$. That they are indeed attained as spectral values, can then be deduced from the index theory arguments of \S\ref{sec:coarse.Landau}. For the Euclidean plane case, the full spectrum can be obtained algebraically, and there is ``no room'' for any continuous spectrum, see \S\ref{sec:Euclidean.Landau}. 
\end{remark}

\subsection{The case of the Euclidean plane}\label{sec:Euclidean.Landau}

In the case that $X = \mathbb{E}$, the Euclidean plane, with constant $\theta \in \RR \setminus \{0\}$, the eigenvalues of the Landau Hamiltonian are the infinite sequence
\begin{equation*}
 \lambda_{m, \theta} = (2m +1)|\theta|,\qquad m=0, 1, 2, \dots,
\end{equation*}
and no continuous spectrum occurs. Since the scalar curvature $R = 0$, we obtain from \eqref{eqn:Lich} that
\begin{equation*}
\slashed{D}_\theta^2=\begin{pmatrix}H_\theta-\theta & 0 \\ 0 & H_\theta+\theta\end{pmatrix} \geq 0.
\end{equation*}
Hence in this case, there is no shift in the $\theta$ parameter and operators with different $\theta$ parameters are unrelated. Using similar methods to the hyperbolic case, one shows that if again $\mathcal{E}_{m, \theta}$ are the eigenspaces to the eigenvalues $\lambda_{m, \theta}$, then
\begin{equation*}
  \mathcal{E}_{0, \theta} = \begin{cases} \ker \slashed{D}_\theta^+ & \text{if}~\theta>0, \\ \ker \slashed{D}_\theta^- & \text{if} ~\theta <0.\end{cases}
\end{equation*}
Moreover, the operators $V_\theta$, respectively $V_\theta^*$, defined by the same formula \eqref{DefinitionV} as in the hyperbolic case, provide unitary isomorphisms $\mathcal{E}_{m, \theta} \cong \mathcal{E}_{0, \theta}$, see Fig.\ \ref{fig:euc.raising.lowering}. This construction is well-known in the physics literature, and a rigorous account can be found in \S7.1.3 of \cite{Thaller}.

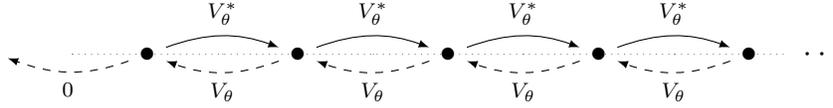
\begin{figure}
\begin{center}
\begin{tikzpicture}
\draw[step=1cm,gray, thin, dotted] (0,0) grid (9.5,0);
\node (0) at (1,0) {$\bullet$};
\node (1) at (3,0) {$\bullet$};
\node (2) at (5,0) {$\bullet$};
\node (3) at (7,0) {$\bullet$};
\node (4) at (9,0) {$\bullet$};
\node(-1) at (-1,0) {};
\node at (10,0) {$\cdots$};
\draw[-latex,dashed] (1) to[out=-160,in=-20] node[midway,font=\scriptsize,below] {$V_\theta$} (0);
\draw[-latex] (0) to[out=20,in=160] node[midway,font=\scriptsize,above] {$V_\theta^*$} (1);
\draw[-latex,dashed] (2) to[out=-160,in=-20] node[midway,font=\scriptsize,below] {$V_\theta$} (1);
\draw[-latex] (1) to[out=20,in=160] node[midway,font=\scriptsize,above] {$V_\theta^*$} (2);
\draw[-latex,dashed] (3) to[out=-160,in=-20] node[midway,font=\scriptsize,below] {$V_\theta$} (2);
\draw[-latex] (2) to[out=20,in=160] node[midway,font=\scriptsize,above] {$V_\theta^*$} (3);
\draw[-latex,dashed] (4) to[out=-160,in=-20] node[midway,font=\scriptsize,below] {$V_\theta$} (3);
\draw[-latex] (3) to[out=20,in=160] node[midway,font=\scriptsize,above] {$V_\theta^*$} (4);
\draw[-latex,dashed] (0) to[out=-160,in=-20] node[midway,font=\scriptsize,below] {$0$} (-1);
\end{tikzpicture}
\end{center}
\caption{Spectrum of Euclidean Landau Hamiltonian $H_\theta$, pictured as a subset of the horizontal dotted line. The Landau levels (labelled by $\bullet$) are isolated and evenly spaced. In the case $\theta >0$ , the operators $V_\theta$ lower the Landau levels, except at $m=0$, where $V_\theta = 0$. 
Similarly, $V_\theta^*$ raises the Landau levels (solid arrows).}\label{fig:euc.raising.lowering}
\end{figure}

\section{The coarse topological invariant of Landau levels}\label{sec:coarse.Landau}

Having identified each Landau eigenspace $\mathcal{E}_{m,\theta}$ as the kernel of a twisted Dirac operator $\slashed{D}_a$ for some suitable $a\in\RR$, it is natural to go further and identify $\mathcal{E}_{m,\theta}$ as the \emph{index} of $\slashed{D}_a$. Because $X$ is a noncompact manifold, it is necessary to use the notion of the \emph{coarse index} for $\slashed{D}_a$, which is a class in $K_0(C^*(X))$, the $K$-theory of the Roe algebra of $X$. 

\subsection{The coarse index}\label{sec:coarse.index}

If $F$ is a Fredholm operator on a Hilbert space $\mathcal{H}$, its \emph{index} is the integer $\dim \ker(F) - \dim \mathrm{coker}(F)$; the significance of this integer is that it is deformation invariant, while the individual dimensions of $\ker(F)$ and $\mathrm{coker}(F)$ are not.

In $K$-theory language, the index of $F$ can be identified with the formal difference $\mathrm{Ind}(F) := [\Pi_{{\rm ker}(F)}]-[\Pi_{{\rm ker}(F^*)}]$ of projections onto kernel and cokernel, which represents a class in $K_0(\mathcal{K}) \cong \ZZ$, the $K_0$-group of the algebra $\mathcal{K} = \mathcal{K}(\mathcal{H})$ of compact operators on a Hilbert space. This index arises naturally as the image of the boundary map in the $K$-theory six-term sequence associated to the short exact sequence
\begin{equation} \label{CompBoundSequence}
0\longrightarrow \mathcal{K}\longrightarrow \mathcal{B}\overset{\pi}{\longrightarrow}\mathcal{B}/\mathcal{K}\longrightarrow 0,
\end{equation}
with $\mathcal{B} = \mathcal{B}(\mathcal{H})$ the bounded operators. That $F$ is Fredholm means that it is invertible modulo compact operators, hence $\pi(F)$ is invertible in  $\mathcal{B}/\mathcal{K}$ and defines a class $[\pi(F)] \in K_1(\mathcal{B}/\mathcal{K})$. The index of $F$ from before is then given by $\partial[\pi(F)]={\rm Ind}(F) \in K_0(\mathcal{K})$, where $\partial:K_1(\mathcal{B}/\mathcal{K})\rightarrow K_0(\mathcal{K})$ is the connecting map of the six-term sequence in $K$-theory associated to \eqref{CompBoundSequence}.

The typical example is the Dirac operator $\slashed{D}$ on an even-dimensional compact spin manifold, which can be written as
\begin{equation*}
\slashed{D}=\begin{pmatrix} 0 & \slashed{D}^- \\ \slashed{D}^+ & 0\end{pmatrix}
\end{equation*}
with respect to the even-odd grading of the spinor bundle. Here the index of $\slashed{D}^+$ turns out to be a topological invariant, which is calculated by the celebrated Atiyah--Singer index theorem\footnote{Since $\slashed{D}$ is self-adjoint, the index of $\slashed{D}^-$ satisfies $\mathrm{Ind}(\slashed{D}^-) = - \mathrm{Ind}(\slashed{D}^+)$ and gives nothing new.}. While $\slashed{D}^+$ is unbounded, the operator $F := \slashed{D}^+ (1+\slashed{D}^-\slashed{D}^+)^{-1/2}$ is bounded and has the same Fredholm index as $\slashed{D}^+$. 

On a complete, but non-compact spin-manifold $X$, the Dirac operator is no longer Fredholm in general. To still extract an invariant, one uses the short exact sequence
\begin{equation} \label{RoeSES}
0\longrightarrow C^*(X)\longrightarrow D^*(X)\overset{\pi}{\longrightarrow}D^*(X)/C^*(X)\longrightarrow 0.
\end{equation}
of Roe algebras, recalled in \S\ref{sec:Roe.algebras}. One shows that the operator $F$, defined by the same formula as before, is invertible modulo $C^*(X)$, hence $\pi(F)$ defines a class in $K_1(D^*(X)/C^*(X))$. The \emph{coarse index of} $\slashed{D}$ as defined in \cite{Roe-book}, \S12.3 of \cite{HigsonRoeBook}, is then
\begin{equation*}
   \mathrm{Ind}(\slashed{D}) := \partial ([\pi(F)]) \in K_0(C^*(X)), 
\end{equation*}
where $\partial: K_1(D^*(X)/C^*(X)) \rightarrow K_0(C^*(X))$ is the boundary map corresponding to the short exact sequence \eqref{RoeSES}.

This index is somewhat abstract so far and does not have a clear interpretation in terms of kernel and cokernel in general. For example, for both $X = \EE$ and $\mathbb{H}$, the spectrum of $\slashed{D} = \slashed{D}_0$ is the whole real line while its kernel is zero. On the other hand, if we took the \emph{twisted} Dirac operator $\slashed{D}_a$ for some $a\neq 0$, zero is isolated in the spectrum (Lemma \ref{lem:isolated.zero}) and one shows that the coarse index is just the formal difference
\begin{equation} \label{CoarseIndexDa}
  \mathrm{Ind}(\slashed{D}_a) = [\Pi_{\ker(\slashed{D}_a^+)}] - [\Pi_{\ker(\slashed{D}_a^-)}]
\end{equation}
similar to the finite-dimensional case, where $\Pi_{\ker(\slashed{D}_a^\pm)} \in C^*(X)$ denotes the orthogonal projection onto $\ker(\slashed{D}_a^\pm)$.

\subsection{Coarse index of Landau levels}

\begin{lemma}\label{lem:coarse.K.groups}
For $X$ the Euclidean or hyperbolic plane, we have 
\begin{equation*}
K_i(X)\cong K_i(C^*(X))\cong\begin{cases} \ZZ,\qquad i=0,\\ 0\;\;\,\qquad i=1,\end{cases}
\end{equation*}
where a generator is $\mathrm{Ind}(\slashed{D})$,  the coarse index of the Dirac operator.
\end{lemma}

\begin{proof}
The coarse Baum--Connes connjecture is verified in these cases, see Cor.\ 8.2, Prop.\ 3.8, Conjecture 6.4 of \cite{HRcoarse}. Namely, with $K_i(X)$ the Kasparov $K$-homology group, the coarse assembly map
\begin{equation*}
\mu:K_i(X)\rightarrow K_i(C^*(X)),
\end{equation*}
is an isomorphism. We remark that a possible definition for the left hand side is 
\begin{equation*}
K_i(X) := K_{1-i}(\Psi^0(X) / \Psi^{-1}(X)) \qquad (i = 0, 1),
\end{equation*}
 where $\Psi^0(X)$ is the algebra of pseudolocal operators on $X$ and $\Psi^{-1}(X)$ is the algebra of locally compact operators on $X$ (see, e.g., \S5 in \cite{HRcoarse}). It is then possible to show that $\Psi^0(X) / \Psi^{-1}(X) = D^*(X)/C^*(X)$ and the assembly map $\mu$ is just the boundary map of the six-term sequence associated to \eqref{RoeSES}.
Now note that the groups $K_i(X)$ depend only on the topology of $X$ (not its coarse geometry), and for the plane, we have $K_0(X)\cong\ZZ$ and $K_1(X)=0$.

An explicit generator for $K_0(X)\cong\ZZ$ is the class $[F]$ (where $F := \slashed{D}^+ (1+\slashed{D}^-\slashed{D}^+)^{-1/2}$) of the standard Dirac operator $-$ this is the ``fundamental class'' of $X$ in the sense of Definition 11.2.10 in \cite{HigsonRoeBook}. It is then a standard fact, \S6 of \cite{HRcoarse}, that the assembly map sends $[F]$ to the coarse index of $\slashed{D}$, i.e., $\mu([F]) = \mathrm{Ind}(\slashed{D})$, whence the claim follows.
\end{proof}

\begin{lemma}\label{lemma:Dirac.index.nontrivial}
Let $X$ be the Euclidean or hyperbolic plane, and $\slashed{D}_a$ the twisted Dirac operator {\normalfont (\S\ref{sec:Lich})}. The class $\mathrm{Ind}(\slashed{D}_a) \in K_0(C^*(X))$ is independent of $a \in \RR$.
\end{lemma}

\begin{proof}
Since $1 + \slashed{D}^-_a \slashed{D}^+_a$ is a positive differential operator of order two with scalar principal symbol, it is well-known that the operator $R_a := (1 + \slashed{D}^-_a \slashed{D}^+_a)^{-1/2}$ is a pseudodifferential operator of order $-1$, with the same principal symbol when restricted to the unit sphere (see, e.g., \cite[\S9]{ShubinBook}). It follows that the operator $F_a = \slashed{D}^+_a R_a$ is an elliptic pseudodifferential operator of order zero with principal symbol on the unit sphere equal to that of $\slashed{D}_a^+$. In particular, the principal symbol of $F_a$ is independent of $a \in \RR$ (since that of $\slashed{D}_a^+$ is independent of $a \in \RR$), and for $a, b \in \RR$, the difference $F_a - F_b$ is a pseudodifferential operator of order $-1$, hence locally compact.

Now since $F_a$ is elliptic, it is invertible modulo $\Psi^{-1}(X)$, hence defines an element $[F_a] \in K_1(\Psi^0(X) / \Psi^{-1}(X)) = K_0(X)$, and we have
\begin{equation*}
  \mathrm{Ind}(\slashed{D}_a) = \mu ( [F_a] )
\end{equation*}
by the proof of Lemma~\ref{lem:coarse.K.groups}. By the considerations before, for any $a, b \in \RR$, we have $[F_a] = [F_b]$ in $\Psi^0(X) / \Psi^{-1}(X)$, whence the claim follows.
\end{proof}

\begin{theorem}\label{thm:nontrivial.coarse.index}
  For either $X = \mathbb{E}$ or $\mathbb{H}$, let $\Pi_{m, \theta}$ be the projection onto the $m$-th Landau eigenspace $\mathcal{E}_{m, \theta}$. Then the class of $\Pi_{m, \theta}$ is a generator of $K_0(C^*(X))$.
\end{theorem}

\begin{proof}
To begin with, observe that $\Pi_{m, \theta}$ is indeed contained in $C^*(X)$: Since the $m$-th Landau level $\lambda_{m, \theta}$ is separated from the rest of the spectrum, there exists a continuous function $\phi$ with $\phi(\lambda_{m, \theta}) = 1$ and $\phi(\lambda) = 0$ for all $\lambda \in \mathrm{spec}(H_\theta) \setminus \{\lambda_{m, \theta}\}$. Hence $\Pi_{m, \theta} = \phi(H_\theta) \in C^*(X)$, by Lemma~\ref{PropFunctionsInCStarX}.

Consider first the case that $X = \mathbb{H}$ and suppose that $\theta > m +\frac{1}{2}$. Let us rewrite the lowering operator $V_{\theta-\frac{1}{2}}$ restricted to the subspace $\mathcal{E}_{m,\theta}$, in terms of $\slashed{D}_{\theta-\frac{1}{2}}$, following Remark \ref{rem:lowering.remark}. Let $\varphi$ be a compactly supported function with
\begin{equation*}
  \varphi(\lambda) = \begin{cases} \mathrm{sign}(\lambda), & \lambda = \pm(\lambda_{m, \theta}-\theta) \\ 0, & \lambda \in \mathrm{spec}(\slashed{D}_{\theta-\frac{1}{2}}) \setminus \{\pm (\lambda_{m, \theta}-\theta)\} \end{cases}
\end{equation*}
and set $v_{m, \theta} = \varphi(\slashed{D})^+$. Then $v_{m, \theta} \in C^*(X)$ by Lemma~\ref{PropFunctionsInCStarX}, and by construction, $v_{m, \theta}$ acts on $\mathcal{E}_{m, \theta}$ just as $V_{\theta-\frac{1}{2}}$ does, in particular, $v_{m, \theta} \mathcal{E}_{m, \theta} = \mathcal{E}_{m-1, \theta-1}$.
More precisely, by Prop.~\ref{prop:induction}, we have
\begin{equation*}
  v_{m, \theta}^* v_{m, \theta} = \Pi_{m, \theta}, \qquad v_{m, \theta} v_{m, \theta}^* = \Pi_{m-1, \theta-1}.
\end{equation*}
Hence $\Pi_{m, \theta}$ and $\Pi_{m-1, \theta-1}$ are Murray-von-Neumann equivalent as projections in $C^*(X)$, hence define the same element in $K$-theory.

Iterating the argument and looking at Corollary~\ref{cor:Landau.Dirac.kernel}, we obtain that $\Pi_{m, \theta}$ is Murray-von-Neumann equivalent to the projection onto the kernel of $\slashed{D}^+_{\theta - m - \frac{1}{2}}$. But since $\slashed{D}_{\theta-m-\frac{1}{2}}^-$ has trivial kernel, formula \eqref{CoarseIndexDa} gives that the $K$-theory class $[{\rm ker}(\slashed{D}_{\theta-m-\frac{1}{2}}^+)]$ is equal to $\mathrm{Ind}(\slashed{D}_{\theta-m-\frac{1}{2}})$, the coarse index of the twisted Dirac operator, which generates $K_0(C^*(X))$ by Lemma~\ref{lemma:Dirac.index.nontrivial}.

If $\theta < -m - \frac{1}{2}$, one obtains with the same reasoning that $\Pi_{m, \theta}$ is Murray-von-Neumann equivalent to the projection onto the kernel of $\slashed{D}^-_{\theta + m + \frac{1}{2}}$. In this case, the kernel of $\slashed{D}^+_{\theta + m + \frac{1}{2}}$ is trivial, hence $[\Pi_{m, \theta}] = - \mathrm{Ind}(\slashed{D}_{\theta + m + \frac{1}{2}})$, which is again a generator of $K_0(C^*(X))$.

The Euclidean case, $X = \mathbb{E}$, is exactly analogous and simpler, with the Landau level lowering operators given in \S\ref{sec:Euclidean.Landau}.
\end{proof}

\begin{lemma}\label{lem:standard.half.plane}
Let $X$ be the Euclidean or hyperbolic plane and let $W \subset X$ be a closed half-plane with boundary a geodesic. Then the boundary map 
\begin{equation*}
\partial_0 : \ZZ \cong K_0(C^*(X)) \longrightarrow K_1(C^*_{W}(\partial W))\cong \ZZ
\end{equation*}
constructed in \S\ref{sec:Roe.algebras} is an isomorphism. In particular, the generator $\mathrm{Ind}(\slashed{D})$ of the left hand side is mapped to a generator of  $K_1(C^*_{W}(\partial W))$.
\end{lemma}

\begin{proof}
In the Euclidean case, we have $W = [0, \infty) \times \RR$ so that $W$ is {\em flasque}, a property of coarse spaces which entails that $K_i(C^*(W)) = 0$ for $i = 0, 1$. The claim then follows from the six-term sequence \eqref{SixTerm1}. 

For $X = \HH$, the half space $W$ is not flasque, so we have to give another argument (which also works for $X=\EE$ above). Here we can use the Mayer--Vietoris sequence \eqref{MayerVietoris} with $W^\prime = \overline{X \setminus W}$. Reflection at $\partial W$ provides an algebra isomorphism $C^*(W) \cong C^*(W^\prime)$. Moreover, it is known that $K_0(C^*_X(\partial W)) \cong K_0(C^*(\RR)) = 0$ and $K_0(C^*_X(\partial W))\cong K_1(C^*(\RR))\cong\ZZ$ (pp.\ 33 of \cite{Roe-coarse-book}), while we saw that that $K_1(C^*(\HH)) = 0, K_0(C^*(\HH))=\ZZ$ (Lemma~\ref{lem:coarse.K.groups}). Putting these groups into the Mayer--Vietoris sequence therefore yields an exact sequence of the form
\begin{equation*}
\begin{tikzcd}
  0 \ar[r] & A_0 \oplus A_0 \ar[r] & \ZZ \ar[r, "\partial_0"] & \ZZ  \ar[r] & A_1 \oplus A_1 \ar[r] & 0,
\end{tikzcd}
\end{equation*}
where $A_i = K_i(C^*(W))$. But one easily checks that no matter what $A_0$ and $A_1$ are, there are no injective group homomorphisms $A_0 \oplus A_0 \rightarrow \ZZ$ and no surjective group homomorphisms $\ZZ \rightarrow A_1 \oplus A_1$. Hence the only solution to this algebraic problem is $A_0 = A_1 = 0$ and $\partial_0$ an isomorphism. Observe that this also proves that $K_i(C^*(W)) = 0$ for $i= 0, 1$.
\end{proof}

\begin{remark}
While the coarse index of $\slashed{D}_a$ is independent of $a$, it is only realised as ``kernel minus cokernel'' when $a\neq 0$. For instance, for all $a>0$, we have 
\begin{equation*}
[{\rm ker}(\slashed{D}_a)]={\rm Ind}(\slashed{D}_a)={\rm Ind}(\slashed{D}_{-a})=-[{\rm ker}(\slashed{D}_{-a})],
\end{equation*}
 and it is precisely the failure of these expressions at $a=0$ which allows for the ``discontinuity'' in the $K$-theory class of the Dirac kernel as $a$ is varied. Physically, changing $a=\theta-\frac{1}{2}>0$ into $-a=\theta+\frac{1}{2}$ corresponds to reversing the magnetic field $\theta\rightarrow-\theta$. Semiclassically, the cyclotron motion of the Landau level eigenstates reverses accordingly, leading to the chiral current flowing in the opposite direction along the boundary. The quantisation of the latter current may be attributed to $\partial_0[\Pi_{m,\theta}]\in K_1(C^*(\partial W))$ (\S6 of \cite{LTcobordism}). The sign change of $[\Pi_{m,\theta}]$ in Theorem \ref{thm:nontrivial.coarse.index} is consistent with these physical considerations.
\end{remark}

\begin{remark} \label{RemarkImperfekt}
Using the invariance of Roe algebra $K$-theory under coarse equivalences, the proof of Lemma~\ref{lem:standard.half.plane} generalizes to subsets $W \subset X$ that are ``imperfect half spaces'' whose boundary $\partial W$ is only ``roughly a geodesic''. By this, we mean a regular closed subset $W$, with regular complement $W^\prime=\overline{X\setminus W}$, for which there exists a connected, complete, totally geodesic, one-dimensional submanifold $N \subset X$ (in other words, $N$ is the image of a geodesic in $X$) such that the following hold: (1) There exists $R>0$ such that $\partial W \subset B_R(N)$ and $N \subset B_R(\partial W)$; (2) There exists $S>0$ such that $\rho(W^\prime) \subset B_{S}(W)$ and $\rho(W) \subset B_{S}(W^\prime)$. Here for $A\subset X$,  $B_R(A)$ denotes the $R$-ball around $A$ and $\rho: X \to X$ is the isometric reflection across $N$. 
%In other words, an imperfect half space is a regular closed subset $W \subset X$ with $d_H(W, W_{\mathrm{per}}) < \infty$ and $d_H(\partial W, \partial W_{\mathrm{per}})<\infty$ for some ``true'' half space $W_{\mathrm{per}}$ with $\partial W_{\mathrm{per}}=N$, where $d_H$ denotes the Hausdorff distance in $X$.
These ensure that $W^\prime$ is coarsely equivalent to $\rho(W)$ and thus to $W$, and that $\partial W=\partial W^\prime=W\cap W^\prime$ is coarsely equivalent to the geodesic $N$.
\end{remark}

\subsection{Gaplessness of half-space hyperbolic Landau Hamiltonians} \label{SectionMainTheorem}

Finally, we explain how the $K$-theoretic non-triviality of the Landau projections, $0\neq [\Pi_{m,\theta}]\in K_0(C^*(X))$, implies our main theorem about the gap-filling phenomenon of Landau Hamiltonians. 

We repeat the setup and then state the main theorem in a more general form. Let $X$ be either the hyperbolic or the Euclidean plane and let $W$ be the closed half-plane lying on one side of a geodesic or, more generally, an ``imperfect half space'' in the sense of Remark~\ref{RemarkImperfekt}. 
Let $H_{\theta, W}$ be the self-adjoint extension of the Landau Hamiltonian $H_\theta$ (initially defined by \eqref{eqn:Landau.Hamiltonian} on $C^\infty_c(\mathring{W})$), which is obtained from imposing either Dirichlet boundary conditions or (if $\partial W$ is sufficiently regular) Neumann boundary conditions or, more generally, a self-adjoint extension satisfying the assumptions stated in Remark~\ref{RemarkGeneralBC}.

\begin{theorem}\label{thm:main.theorem}
%Let $X$ be the hyperbolic or Euclidean plane and let $W$ be the closed half-plane lying on one side of a geodesic. 
%Let $H_{\theta, W}$ be a Landau Hamiltonian on $W$, i.e., a self-adjoint extension of $H_\theta$, initially defined by \eqref{eqn:Landau.Hamiltonian} on $C^\infty_c(\mathring{W})$, which is non-negative and satisfies the property that the corresponding wave operator has finite propagation speed (e.g., obtained from imposing Dirichlet or Neumann boundary conditions).
The spectrum of $H_{\theta,W}$ has no gaps above the lowest Landau level $\lambda_{0,\theta}=|\theta |$. 
\end{theorem}

\begin{proof}
For the Landau projections $\Pi_{m,\theta}$, we know from Theorem \ref{thm:nontrivial.coarse.index} and Lemma \ref{lem:standard.half.plane} that 
\begin{equation*}
\partial_0[\oplus_{m=0}^{m^\prime}\Pi_{m,\theta}]=(m^\prime+1)\partial_0[{\rm Ind}(\slashed{D})]\neq 0
\end{equation*} for every $m^\prime\in\NN$ in the Euclidean case, and every $m^\prime =0, 1, \ldots, m_{\rm max}$ in the hyperbolic case. Applying this to Theorem \ref{thm:gap.filling}, we deduce that either the gap below the lowest Landau level, or the gap above the $m^\prime$-th Landau level must be filled when passing to $H_{\theta,W}$. By assumption, $H_{\theta,W}$ is bounded below, so it is the latter gap that is filled, and this holds for every $m^\prime$. In view of Remark \ref{rem:equivariant.advantage}, we also have $[\frac{1}{4},\infty)\subset{\rm spec}(H_\theta)\subset {\rm spec}(H_{\theta,W})$, thus no new gaps are introduced into the (already connected) continuous spectrum of $H_\theta$, when passing to $H_{\theta,W}$. Thus, $H_{\theta,W}$ has no spectral gaps at all above the lowest Landau level $\lambda_{0,\theta}=|\theta|$.
\end{proof}

%{\bf General half-planes.} In our coarse geometry approach, it is manifestly clear that we do not need to choose $W, W^\prime$ to be standard half-planes intersecting on a geodesic. For example, we could take $W, W^\prime$ quasi-isometric, or even coarsely equivalent to the standard half-plane (with intersection coarsely equivalent to $\RR$). Lemma \ref{lem:standard.half.plane} still holds, and when fed into the gap-filling mechanism, Theorem \ref{thm:gap.filling}, leads to Theorem \ref{thm:main.theorem} for more general choices of $W$.

\begin{remark}\label{rem:higher.dimensions}
While we have focused on two-dimensional $X$ in order to answer a concrete open question about half-space hyperbolic Landau Hamiltonians, our methods also work for higher-dimensional $X$. As a simple Euclidean example, Landau levels also arise for magnetic Lapacians on $X=\RR^{2n}$ \cite{Goffeng}. We may take, in the first instance, $W$ to be a half-space with $\partial W\cong\RR^{2n-1}$, then flasqueness of $W$ leads to the Mayer--Vietoris boundary map $\partial_0:K_0(C^*(\RR^{2n}))\rightarrow K_1(C^*(\RR^{2n-1}))$ being an isomorphism (Lemma \ref{lem:standard.half.plane} again holds). Then the gap-filling Theorem \ref{thm:gap.filling} again applies. We mention that such generalisations to higher-dimensional situations were also suggested in \cite{Yuezhao}. 
\end{remark}

\begin{remark} \label{RemarkCoarse}
The assumptions on ``imperfect half spaces'' $W$ from Remark~\ref{RemarkImperfekt} imply that $W$ is in fact quasi-isometric to a standard half plane, 
a much stronger statement than coarse equivalence. However, using our coarse index theory methods \cite{LTcobordism}, Thm.~\ref{thm:main.theorem} may be generalized beyond this class of imperfect half planes. 
\end{remark}

\section*{Ackowledgements}
The authors thank  U.\ Bunke, N.\ Higson,  Y.\ Li, and R.\ Meyer for their helpful correspondence. M.L.\ thanks the SFB 1085 ``Higher Invariants'' for support. G.C.T.\ acknowledges support from Australian Research Council DP200100729, and the University of Adelaide for hosting him.

\bibliography{literature}

\begin{thebibliography}{10}

\bibitem{Benedict}
M.~Benedict and B.~Moln\'{a}r.
\newblock Algebraic construction of the coherent states of the {M}orse
  potential based on supersymmetric quantum mechanics.
\newblock {\em Physical Review A}, 60(3):1737--1740, 1999.

\bibitem{BerlineGetzlerVergne}
N.~Berline, E.~Getzler, and M.~Vergne.
\newblock {\em Heat kernels and {D}irac operators}.
\newblock Springer, Berlin, 1992.

\bibitem{QuantumSimulation}
I.~Boettcher, P.~Bienias, R.~Belyansky, A.~J. Koll\'ar, and A.~V. Gorshkov.
\newblock Quantum simulation of hyperbolic space with circuit quantum
  electrodynamics: From graphs to geometry.
\newblock {\em Phys. Rev. A}, 102:032208, Sep 2020.

\bibitem{BMR}
V.~Bruneau, P.~Miranda, and G.~Raikov.
\newblock Dirichlet and {N}eumann eigenvalues for half-plane magnetic
  {H}amiltonians.
\newblock {\em Reviews in Mathematical Physics}, 26:1450003, 2014.

\bibitem{CHMM}
A.~Carey, K.~Hannabuss, V.~Mathai, and P.~McCann.
\newblock Quantum {H}all effect on the hyperbolic plane.
\newblock {\em Communications in Mathematical Physics}, 190:629--673, 1998.

\bibitem{Comtet}
A.~Comtet and P.~Houston.
\newblock Effective action on the hyperbolic plane in a constant external
  field.
\newblock {\em Journal of Mathematical Physics}, 26(1):185, 1985.

\bibitem{Pule}
S.~De~Bi{\`e}vre and J.~Pul{\'e}.
\newblock Propagating edge states for a magnetic {H}amiltonian.
\newblock {\em Mathematical Physics Electronic Journal}, 5:33--55, 2002.

\bibitem{EwertMeyer}
E.~E. Ewert and R.~Meyer.
\newblock Coarse geometry and topological phases.
\newblock {\em Communications in Mathematical Physics}, 366(3):1069--1098,
  2019.

\bibitem{Gendenshtein}
L.~Gendenshtein.
\newblock Derivation of exact spectra of the {S}chr\"{o}dinger equation by
  means of supersymmetry.
\newblock {\em JETP Letters}, 38(6):356--359, 1983.

\bibitem{Goffeng}
M.~Goffeng.
\newblock Index formulas and charge deficiencies on the {L}andau levels.
\newblock {\em Journal of Mathematical Physics}, 51:023509, 2010.

\bibitem{HRcoarse}
N.~Higson and J.~Roe.
\newblock On the coarse {B}aum--{C}onnes conjecture.
\newblock In S.~Ferry, A.~Ranicki, and J.~Rosenberg, editors, {\em Novikov
  conjectures, index theorems, and rigidity, Vol. 2}, number 227 in London
  Math. Soc. Lect. Notes, pages 227--254. Cambridge Univ. Press, 1995.

\bibitem{HigsonRoeBook}
N.~Higson and J.~Roe.
\newblock {\em Analytic {$K$}-homology}.
\newblock Oxford University Press, Oxford, 2000.

\bibitem{HRY}
N.~Higson, J.~Roe, and G.~Yu.
\newblock A coarse {M}ayer--{V}ietoris principle.
\newblock {\em Math. Proc. Camb. Phil. Soc.}, 114:85--97, 1993.

\bibitem{Inahama2}
I.~Inahama and S.~Shirai.
\newblock The essential spectrum of {S}chr\"{o}dinger operators with
  asymptotically constant magnetic fields on the {P}oincar\'{e} upper-half
  plane.
\newblock {\em Journal of Mathematical Physics}, 44(1):89--106, 2003.

\bibitem{Kollar}
A.~Koll\'{a}r, M.~Fitzpatrick, and A.~Houck.
\newblock Hyperbolic lattices in circuit quantum electrodynamics.
\newblock {\em Nature}, 571:45--50, 2019.

\bibitem{Landau}
L.~Landau.
\newblock {D}iamagnetismus der {M}etalle.
\newblock {\em Zeitschrift f{\"{u}}r Physik}, 64:629--637, 1930.

\bibitem{Yuezhao}
Y.~Li.
\newblock Coarse {M}ayer-{V}ietoris sequence and {B}ulk-{E}dge
  {C}orrespondence.
\newblock Talk at G{\"{o}}ttingen {S}eminar {N}oncommutative {G}eometry,
  \url{https://researchseminars.org/talk/GoettingenNCG/6}, 2020.

\bibitem{LTcobordism}
M.~Ludewig and G.~C. Thiang.
\newblock Cobordism invariance of topological edge-following states.
\newblock \href{https://arxiv.org/abs/2001.08339}{arXiv:2001.08339}.

\bibitem{MThyperbolic}
V.~Mathai and G.~C. Thiang.
\newblock Topological phases on the hyperbolic plane: fractional bulk-boundary
  correspondence.
\newblock {\em Adv. Theor. Math. Phys.}, 23(3):803--840, 2019.

\bibitem{McKean}
H.~McKean.
\newblock An upper bound to the spectrum of {$\Delta$} on a manifold of
  negative curvature.
\newblock {\em J. Diff. Geom.}, 4:359--366, 1970.

\bibitem{Moller}
M.~Moller.
\newblock On the essential spectrum of a class of operators in {H}ilbert space.
\newblock {\em Math. Nachr.}, 194:185--196, 1998.

\bibitem{Roe-book}
J.~Roe.
\newblock {\em Coarse cohomology and index theory on complete Riemannian
  manifolds}.
\newblock Number 497 in Mem. Am. Math. Soc. Amer. Math. Soc., 1993.

\bibitem{Roe-coarse-book}
J.~Roe.
\newblock {\em Index theory, coarse geometry, and topology of manifolds},
  volume~90.
\newblock American Mathematical Soc., 1996.

\bibitem{ShubinBook}
M.~A. Shubin.
\newblock {\em Pseudodifferential operators and spectral theory}.
\newblock Springer-Verlag, Berlin, second edition, 2001.

\bibitem{Taylor}
M.~E. Taylor.
\newblock {\em Partial differential equations {I}. {B}asic theory}, volume 115
  of {\em Applied Mathematical Sciences}.
\newblock Springer, New York, second edition, 2011.

\bibitem{Thaller}
B.~Thaller.
\newblock {\em The Dirac Equation}.
\newblock Springer-Verlag, 1992.

\end{thebibliography}

\end{document}